\documentclass[10pt,journal,romanappendices]{IEEEtran}
\usepackage{amsmath, amssymb, amsfonts}
\usepackage{amsthm}

\newtheoremstyle{psh}
  {0pt}              % Space above
  {0pt}              % Space below
  {\itshape}         % Italic body
  {\parindent}       % Indented heading
  {\bfseries}        % Bold heading
  {.}                % Punctuation after heading
  {0.5em}            % Space after heading
  {\thmname{#1}\thmnumber{ #2}\thmnote{ (#3)}}

\theoremstyle{psh}
\newtheorem{proposition}{Proposition}

\usepackage{graphicx}
\usepackage{bbm}
\usepackage{mathrsfs}  
\usepackage{booktabs}
\usepackage{array}
\usepackage{placeins}
\usepackage{color}
\usepackage{stfloats}
\usepackage{capt-of}

\allowdisplaybreaks

\begin{document}
\raggedbottom

\title{Pumped-Storage Hydropower Scheduling: An Event-based Convex-hull Approach}

\author{Bo Yang,~\IEEEmembership{Member,~IEEE,} Kai Pan,~\IEEEmembership{Member,~IEEE,} Mohammad Reza Hesamzadeh,~\IEEEmembership{Senior Member,~IEEE}
\thanks{B. Yang is with The Hong Kong University of Science and Technology, K. Pan is with The Hong Kong Polytechnic University, and M. R. Hesamzadeh is with KTH Royal Institute of Technology.}
\vspace{-.6cm}}

\maketitle

\begin{abstract}
Pumped-storage hydropower scheduling couples discrete operating-mode transitions with continuous dispatch decisions through intertemporal reservoir and ramping constraints.
Existing convexification methods do not simultaneously accommodate multiple operating modes, persistent continuous reservoir and ramping states, and interval costs that depend on both incoming and outgoing states of the operating modes.
We develop an event-based dynamic programming formulation in which mode transitions define variable-length events and mode-specific linear programs optimize continuous within-event operation.
We prove that this formulation is equivalent to the conventional time-indexed mixed-integer linear program.
Enumerating the finite set of discrete event paths and scaling their boundary-coupled polyhedra then yields an exact finite-dimensional convex-hull LP without state discretization. 
For tractability, we develop a finite-grid formulation that discretizes reservoir and ramping states only at event boundaries while retaining continuous within-event trajectories. 
Its feasible set contains that of the corresponding stagewise discretization.
Moreover, the resulting LP provides an exact convex-hull representation of the grid-restricted model and admits an error bound linear in the grid resolutions.
An event-based branch-and-bound method provides an independent computational verification of the continuous-state reformulation.
Standalone experiments verify the theoretical results, while tests on a modified IEEE 24-bus Reliability Test System show that the finite-grid LP substantially reduces computation time with a negligible increase in system cost.
\end{abstract}

\vspace{-0.1cm}


\vspace{-0.3cm}

\section{Introduction}
\subsection{Background and Motivation}
\IEEEPARstart{T}{he} increasing penetration of renewable generation has increased the value of resources that can shift energy across time and respond flexibly to changing system conditions.
Pumped-storage hydropower (PSH) is a mature technology for utility-scale, long-duration energy storage and supports energy arbitrage and various system services \cite{Ma2022,IHA2021,IEA2021,huang2021multistage}.
Its operational flexibility is further enhanced by technologies such as hydraulic short-circuit (HSC) operation, which enables simultaneous pumping and generation \cite{chazarra2016optimal,Chazarra2018,gerini2024optimal}.

The resulting scheduling problem combines discrete mode transitions with continuous operating decisions and strong intertemporal coupling.
In a deterministic baseline model, the unit switches among generating, pumping, HSC, and offline modes, while its dispatch variables, reservoir level, and ramping boundary state remain continuous.
The reservoir level links decisions across all modes through the cumulative water balance, while ramping links consecutive generating decisions through the output inherited from the preceding stage.

A conventional time-indexed mixed-integer linear program (MILP) can represent these restrictions, but couples discrete mode transitions with continuous dispatch, reservoir, and ramping trajectories at every stage. 
Although adequate for standalone scheduling, such models can be computationally burdensome when repeatedly solved to coordinate PSH units with other generators under system-level network constraints.
Existing dynamic programming (DP) convexifications do not jointly accommodate multiple PSH operating modes, continuous reservoir and ramping states carried across modes, and mode-level operating costs that depend on the incoming and outgoing values of these states \cite{zhai2010optimization,guan2018polynomial,Qu2023,Xiao2025,mu2026tackling}. 
Finite-state DPs can efficiently solve discretized hydro scheduling problems with pre-computed transition costs \cite{perez2010optimal}.
However, a scalar-cost implementation conceals the within-event variables, while conventional stagewise discretization can exclude feasible trajectories whose event boundaries lie on the grids but whose intermediate states do not.
These limitations motivate a unit-level convex hull formulation that captures all these features with boundary-conditioned events.

We develop an event-based DP formulation in which each event keeps the unit in one operating mode over a consecutive interval. An event specifies its ending stage, successor mode, and outgoing reservoir and ramping boundaries, while a mode-specific linear program (LP) determines its feasible dispatch trajectory and minimum operating cost within the event. We prove a correspondence between feasible time-indexed schedules and event trajectories, establishing that the resulting continuous-state Bellman equation is an exact reformulation of the original PSH MILP.

The event-based reformulation yields an exact convex-hull construction for the PSH MILP. 
Although event boundaries include continuous reservoir level and ramping boundary state, the set of
discrete event paths is finite.
For each path, we link adjacent events through shared boundary variables and apply perspective scaling to the resulting pathwise polyhedron. 
This construction yields a finite-dimensional, though potentially exponentially large, extended LP whose projection onto the original time-indexed variables equals the convex hull of the MILP feasible region. 
Its optimal value therefore coincides with those of the continuous-state DP and the original MILP.

Since the exact convex hull LP may contain exponentially many event paths, we develop two complementary implementations.
Restricting only the boundary states to finite grids yields a computationally
efficient arc-embedded LP that exactly convexifies the grid-restricted event model and admits a discretization-error bound linear in the grid resolutions.
Because the within-event trajectories remain continuous, this model is much less restrictive than the corresponding stagewise discretization model.
Retaining the continuous boundaries yields an exact branch and bound (B\&B) method that branches only on discrete event decisions and uses McCormick-based LP relaxations for bounding. This method is used to verify the continuous-state reformulation.

Numerical experiments validate the theoretical results and demonstrate the computational value of the finite-grid formulation. 
In standalone tests with a single PSH unit, the finite-grid LP exactly reproduces the grid-restricted DP optimum, approaches the continuous-state benchmark under grid refinement, and the event-based B\&B method matches the time-indexed MILP optimum. 
In addition, despite its finer ramping grid, stagewise discretization yields a \(30.75\%\) gap, compared with \(3.867\%\) for our event-based discretized LP at comparable runtimes.
In the 336-stage IEEE 24-bus system test, finite-grid PSH pricing reduces total and pricing times by factors of \(12.9\) and \(41.7\), respectively, while increasing system cost by only \(0.01\%\).
In the enlarged 672-stage instance, the finite-grid LP converges in 10 minutes, whereas the MILP does not converge within the two-hour limit.

\subsection{Related Literature and Contributions}

PSH scheduling has been studied through hydrothermal coordination, deterministic, stochastic, and chance-constrained scheduling, energy-and-reserve co-optimization, detailed mixed-integer formulations, and data-driven methods \cite{chazarra2016optimal,Chazarra2018,gerini2024optimal,guan1994optimization,lu2004pumped,Huang2022,liu2021secured,Toufani2023,wu2025day,favaro2024neural, toubeau2019chance}. 
These studies enrich the physical, market, and uncertainty representations of PSH operation. 
We address a complementary structural question concerning the exact reformulation and convexification of the multi-mode PSH scheduling problem with persistent reservoir and ramping states.

A related stream of literature develops DP, network-flow, and extended formulations for unit commitment and other multistage scheduling problems \cite{zhai2010optimization,guan2018polynomial,Xiao2025,mu2026tackling,perez2010optimal,MartinRardinCampbell1990,deLimaEtAl2022, bruninx2015coupling}. These formulations typically rely on finite discrete states or scalar arc costs indexed by commitment decisions and therefore do not capture PSH schedules in which reservoir dynamics couple successive events and a within-event operating problem determines each event’s feasibility and cost from its continuous incoming and outgoing reservoir and ramping states. A discretized DP could evaluate each event problem and optimize the resulting finite graph, but it typically requires precomputing arc costs, which conceals the within-event variables and does not yield a convex-hull formulation of the complete operating model. By retaining these variables, we construct primal PSH event polyhedra jointly conditioned on incoming and outgoing reservoir and ramping states. 
This polyhedral structure permits perspective scaling to yield both grid-restricted and continuous-state convexifications.

Unlike stagewise discretization, our finite-grid LP restricts reservoir and ramping states only at variable-length event boundaries. It therefore contains the stagewise-grid feasible set and yields an objective no worse.
We are unaware of any prior PSH convexification that enforces inter-event reservoir constraints through boundary-conditioned polyhedra and discretizes reservoir and ramping states only at endogenous event boundaries.
%To our knowledge, no prior short-term PSH unit-commitment formulation combines modeling variable-length events conditioned on continuous reservoir and ramping boundary states with boundary-only discretization that preserves continuous within-event trajectories.

Within the PSH literature, our previous work~\cite{Qu2023} provides the closest interval-based formulation. 
In that work, the DP arc costs, however, are not conditioned on arc-specific incoming and outgoing reservoir states. We instead make the reservoir and ramping boundaries explicit components of every event state and decision. Each event consequently represents a self-contained feasible transition; conditional on a discrete event path, the reservoir, ramping, and operating decisions remain disaggregated and form a boundary-coupled polyhedron. This structure enables both the continuous- and discrete-state convexifications. It also supports direct recovery of the dispatch, reservoir, and ramping trajectories.

Taken together, this paper makes three contributions. 
\begin{enumerate}
    \item First, we formulate an event-based DP that is exactly equivalent to the time-indexed PSH MILP. Event decisions specify the event duration, successor mode, and outgoing reservoir and ramping boundaries, while mode-specific LPs determine the feasible within-event operation and its minimum cost.
    \item Second, we construct a boundary-coupled polyhedron for every discrete event path while retaining continuous reservoir, ramping, and operating decisions. 
    Scaling these pathwise polyhedra yields a finite-dimensional, though potentially exponential-size, LP whose projection is exactly the convex hull of the original time-indexed feasible set.
    \item Third, we develop a tractable finite-grid implementation and an exact continuous-state verification method.
    Restricting only the reservoir and ramping states transmitted across event boundaries to finite grids yields an arc-embedded LP that exactly convexifies the grid-restricted event model, while within-event operation remains continuous. 
    With common grids, the event-boundary formulation attains an objective value no worse than the conventional stagewise discretization model.
    It also admits a controllable discretization-error bound.
    Retaining the continuous boundaries yields an exact event-based B\&B method that provides an independent verification benchmark for the original model.
    %Standalone experiments verify the theoretical results, while system-level experiments demonstrate the computational value of the finite-grid formulation.
\end{enumerate}

The remainder of the paper presents the time-indexed formulation in Section~\ref{sec:MILP}, the event-based formulation and solution methods in Sections~\ref{sec:event_model}--\ref{sec:ebb}, and numerical results in Section~\ref{sec:num}, followed by conclusions in Section~\ref{sec:conc}. All proofs are provided in the appendix.

\section{Time-Indexed Formulation}
\label{sec:MILP}
We study the operation of a PSH facility with generating and pumping equipment and a reservoir. The facility provides grid-scale energy storage. 
It can generate power by releasing water from the reservoir through the turbine, store energy by pumping water into the reservoir using power purchased from the grid, or perform both operations simultaneously. 
The PSH operator aims to minimize the total operating cost.

We consider a planning horizon of $T$ stages, indexed by $t\in\mathcal{T}:=\{1,\ldots,T\}$. 
At stage $t$, the reservoir is characterized by the reservoir level $M_t$, spillage $\Lambda_t$, and natural inflow $R_t$. 
The reservoir capacity is $\overline{M}$, and the spillage \(\Lambda_t\) and natural inflow \(R_t\) are exogenous inputs.
We consider four operating modes $\{\mathsf G,\mathsf P,\mathsf{SC},\mathsf O\}$,
where the four symbols denote generating, pumping, HSC operation, and offline, respectively. 
In generating mode, the unit produces $H_{t}^{\mathsf{O}}$ units of electricity and releases $\mu_t H_{t}^{\mathsf{O}}$ units of water, where $\mu_t$ is the generation water-conversion coefficient.
In pumping mode, the unit consumes $H_{t}^{\mathsf{I}}$ units of electricity and pumps $\alpha_t H_{t}^{\mathsf{I}}$ units of water into the reservoir, where $\alpha_t$ is the pumping efficiency coefficient.
Under HSC operation, the turbine and pump operate simultaneously, so both \(H_t^{\mathsf O}\) and \(H_t^{\mathsf I}\) may be positive at the same stage.
Generation output is constrained to
\([\underline{B}_t^{O},\,\bar{B}_t^{O}]\),
and pumping power consumption is constrained to
\([\underline{B}_t^{I},\,\bar{B}_t^{I}]\).
Ramping rates are limited by $\overline{V}$ throughout the generating process. 
Furthermore, the unit must remain online for at least $L$ consecutive stages and offline for at least $l$ consecutive stages before switching between online and offline.

In addition, we introduce \(SU_t\) and \(SD_t\) to represent start-up and shut-down costs at stage \(t\). 
For brevity, we treat them as known deterministic parameters.
We use binary variables \(y_t^{\mathsf G}\), \(y_t^{\mathsf P}\), and $y_t^{\mathsf{SC}}$ to indicate the unit's operating mode at stage \(t\). 
For instance, $y_{t}^{\mathsf{G}}=1$ indicates that the plant operates in generating mode and $y_{t}^{\mathsf{G}}=0$ otherwise.
Define $y_t := y_t^{\mathsf G} +y_t^{\mathsf P} +y_t^{\mathsf{SC}}$, $g_t := y_t^{\mathsf G} +y_t^{\mathsf{SC}}$, and $p_t := y_t^{\mathsf P} +y_t^{\mathsf{SC}}$. 
Thus, \(y_t\) indicates whether the unit is online, while \(g_t\) and \(p_t\) activate turbine output and pumping input, respectively. 
In HSC mode, \(g_t=p_t=1\), so both operations are permitted.
We also introduce binary variables \(u_t\) and \(d_t\), where \(u_t=1\) if the unit starts up at stage \(t\) and \(d_t=1\) if it shuts down.

Let \(h_t(\cdot)\) denote the hydro operating cost at stage \(t\), including convex piecewise-linear generation and pumping costs, net electricity payments at price \(\lambda_t\), and the opportunity value of water. Its epigraph representation is given in Section \ref{sec:event_lp_modules}. Let \(q\) collect the dispatch variables, reservoir levels, operating-mode indicators, and startup and shutdown variables over the planning horizon. The objective is to minimize the total cost:
% \begin{align}
% \label{eq:Obj}
% \min_{q}\sum_{t=1}^{T}[SU_tu_t + SD_td_t + h_t ( H_t^{\mathsf O}, H_t^{\mathsf I}, y_t^{\mathsf G}, y_t^{\mathsf P}, y_t^{\mathsf{SC}},\lambda_t)].
% \end{align}
\begin{IEEEeqnarray}{l}
\label{eq:objective}
\min_q \sum_{t=1}^{T}
\bigl[SU_tu_t + SD_td_t
+ h_t(H_t^O,H_t^I,y_t^G,y_t^P,y_t^{SC},\lambda_t)\bigr].
\IEEEeqnarraynumspace
\end{IEEEeqnarray}
The mode transition constraints and minimum up/down time constraints are
\begin{subequations}
\begin{align}
\sum_{\tau=t-L+1}^{t} u_{\tau}&\leq y_{t},\quad t\in\{L,\dots,T\},\label{eq:min-up}\\
\sum_{\tau=t-l+1}^{t} d_{\tau}&\leq 1-y_t,\quad t\in\{l, \dots,T\},\label{eq:min-down}\\
%\sum_{\tau=t-l+1}^{t} d_{\tau}&\leq 1 - y_{t-l},\quad t\in\{l, \dots,T\},\label{eq:min-down}\\
%y_{t} - y_{t-1} &\leq u_{t},\quad t\in\mathcal{T},\label{eq:swith}\\
%y_{t-1} - y_{t} &\leq d_{t},\quad t\in\mathcal{T},\label{eq:swith_1}\\
y_t-y_{t-1}&=u_t-d_t, \quad t\in\mathcal T, \label{eq:transition}\\
u_t+d_t&\le 1, \quad t\in\mathcal T, \label{eq:transition-exclusive}\\
y_t^{\mathsf G} + y_t^{\mathsf P} + y_t^{\mathsf{SC}} &= y_t, \quad t\in\mathcal T,\label{eq:status}\\
u_t, d_t, y_t^{\mathsf G}, y_t^{\mathsf P}, y_t^{\mathsf{SC}}, y_t &\in\{0,1\},\quad t\in\mathcal T.
\end{align}
\end{subequations}
Constraints~\eqref{eq:min-up}--\eqref{eq:min-down} enforce minimum up/down times, while \eqref{eq:transition}--\eqref{eq:transition-exclusive} record online/offline transitions.
Constraint~\eqref{eq:status} defines the online status and, because \(y_t\) is binary, permits at most one active mode.

The reservoir dynamics are captured by the following balance constraint for $t\in\mathcal{T}$:
\begin{align}
&M_{t+1} = M_{t} -\mu_tH_{t}^{\mathsf{O}}-\Lambda_{t} + R_{t} + \alpha_{t}H_{t}^{\mathsf{I}}.
\end{align}
This equation enforces conservation of water in the reservoir from stage $t$ to $t+1$. 
%The next period storage $M_{t+1}$ equals the current storage $M_{t}$, minus turbine discharge $\mu_t H_{t}^{\mathsf{O}}$ and spillage $\Lambda_{t}$, plus natural inflow $R_{t}$ and the effective pumped inflow $\alpha_{t}H_{t}^{\mathsf{I}}$. 
Note that the spillage $\Lambda_{t}$ and the natural inflow $R_{t}$ are known constants, as we consider a deterministic setting. The boundary conditions are given by
\begin{subequations}
\begin{align}
%& y_{1} = 0, \label{eq:init-off} \\
%& y_{t}^{\mathsf{G}} \underline{B}_{t}^{\mathsf{O}} \leq H_{t}^{\mathsf{O}} \leq y_{t}^{\mathsf{G}} \bar{B}_{t}^{\mathsf{O}}, \label{eq:gen-bounds} \\
%& y_{t}^{\mathsf{P}} \underline{B}_{t}^{\mathsf{I}} \leq H_{t}^{\mathsf{I}} \leq y_{t}^{\mathsf{P}} \bar{B}_{t}^{\mathsf{I}}, \label{eq:pump-bounds} \\
&g_t\underline B_t^{\mathsf O} \le H_t^{\mathsf O} \le g_t\bar B_t^{\mathsf O}, \quad\ t\in\mathcal{T}\label{eq:gen-bounds}\\
&p_t\underline B_t^{\mathsf I} \le H_t^{\mathsf I} \le p_t\bar B_t^{\mathsf I}, \quad t\in\mathcal{T} \label{eq:pump-bounds}\\
&-\overline V g_{t-1} \leq H_t^{\mathsf O}-H_{t-1}^{\mathsf O} \le \overline V g_t, \quad t\in\mathcal{T}\label{eq:ramp-up}\\
%&H_{t-1}^{\mathsf O}-H_t^{\mathsf O} \le \overline V g_{t-1}, \label{eq:ramp-down}\\
%&-\overline{V}y_{t-1}^{\mathsf{G}} \leq H_{t}^{\mathsf{O}} - H_{t-1}^{\mathsf{O}} \leq \overline{V}y_{t}^{\mathsf{G}}, \label{eq:ramp-up} \\
%& H_{t-1}^{\mathsf{O}} - H_{t}^{\mathsf{O}} \leq \overline{V}y_{t-1}^{\mathsf{G}}, \label{eq:ramp-down} \\
& M_{1} = M^{\text{Initial}}, \label{eq:init-level} \\
& 0 \leq M_{t} \leq \overline{M},\quad t\in\{1,\ldots,T+1\}. \label{eq:reservoir-capacity}
\end{align}
\end{subequations}
We assume the unit is offline at the initial stage by setting $y_0=g_0=H_0^{\mathsf O}=y_{1}=0$.
Constraints~\eqref{eq:gen-bounds} and~\eqref{eq:pump-bounds} activate turbine output in modes \(\mathsf G\) and \(\mathsf{SC}\) and pumping input in modes \(\mathsf P\) and \(\mathsf{SC}\), respectively.
Constraint~\eqref{eq:ramp-up} imposes turbine ramping limits across consecutive stages, including transitions into and out of generating and HSC modes.
\eqref{eq:init-level} sets the initial reservoir level. \eqref{eq:reservoir-capacity} enforces the storage capacity limits at every stage.

\section{Event-Based Modeling}
\label{sec:event_model}
This section reformulates the single-unit PSH model in Section \ref{sec:MILP} as an event-based DP. Each event is a consecutive interval of operation in a fixed mode. Its incoming boundary provides all state information needed to optimize the continuous within-event operation, while its outgoing boundary becomes the incoming state of the next event.

\subsection{Initial State of an Event}
\label{sec:event_state}
At each stage $t\in\mathcal T$, the unit is in one of four modes $x_t\in\mathcal X:=\{\mathsf G,\mathsf P, \mathsf{SC}, \mathsf O\}$.
The state at the beginning of an event includes the operating mode $x_t$, the reservoir level $M_t$, a ramping boundary variable, and a counter for the minimum up/down-time requirement. 
For \(t\geq2\), the ramping boundary \(\bar H_t\) records the previous generation output, i.e., $\bar H_t :=H_{t-1}^{\mathsf O}$ if $x_{t-1}\in\{\mathsf G,\mathsf{SC}\}$; and $0$, if $x_{t-1}\in\{\mathsf P,\mathsf O\}$.
% \begin{align}
% \bar H_t := \begin{cases}
% H_{t-1}^{\mathsf O}, &x_{t-1}\in\{\mathsf G,\mathsf{SC}\},\\
% 0, &x_{t-1}\in\{\mathsf P,\mathsf O\}.
% \end{cases}
% \label{eq:event-ramping-state}
% \end{align}
%At the initial stage, we set \(\bar H_1=0\). 
The counter \(\tau_t\in\{0,\ldots,\tau_{\max}\}\), where \(\tau_{\max}:=\max\{L-1,l-1\}\), records how many stages after \(t\) the unit must retain its current online/offline status before switching between an online mode and \(\mathsf O\).
Hence, the initial state of an event is $s_t:=(x_t,M_t,\bar H_t,\tau_t)\in\mathcal S_t$, where $\mathcal S_t:=\mathcal X\times [0,\overline{M}]\times [0,\overline{H}^{\mathsf O}]\times\{0,\ldots,\tau_{\max}\}$ and $\overline{H}^{\mathsf O}$ is a uniform upper bound on generation output, e.g., $\overline{H}^{\mathsf O}:=\max_{t\in\mathcal T}\bar{B}_t^{\mathsf O}$. 
At $t=1$, we have $s_1=(\mathsf O, M^{\text{Initial}},0,0)$.

\subsection{Events as Building Blocks}
\label{sec:event_blocks}
An event starting at stage $t$ and state $s_t$ is represented by the decision tuple $e_t=(j,x^\dagger, M_j,\bar{H}_j)$, where $j\in\{t+1,\ldots,T+1\}$ is the next decision stage; $x^\dagger\in\mathcal{X}$ is the mode entered at stage $j$; $M_j$ is the reservoir level at stage $j$; and $\bar{H}_j$ is the ramping boundary inherited by the next event.
The event indicates that the plant will remain in mode $x_t$ over stages $\{t,\ldots,j-1\}$ and then enter mode $x^\dagger$ at stage $j$.
The candidate event set is therefore
\begingroup
\allowdisplaybreaks
\begin{align*}
\mathcal E(s_t)
:=
\Big\{
&(j,x^\dagger,M_j,\bar H_j):
j\in\{t+1,\ldots,T+1\},\
x^\dagger\in\mathcal X,\\
&M_j\in[0,\overline M],\
\bar H_j\in[0,\overline H^{\mathsf O}],\\
&j\ge t+\tau_t+1
\text{ if }
x_t\in\{\mathsf G,\mathsf P,\mathsf{SC}\},\
x^\dagger=\mathsf O,\\
&j\ge t+\tau_t+1
\text{ if }
x_t=\mathsf O,\
x^\dagger\in\{\mathsf G,\mathsf P,\mathsf{SC}\}
\Big\}.
\end{align*}
\endgroup
Thus, transitions among \(\mathsf G\), \(\mathsf P\), and \(\mathsf{SC}\) may occur after one stage, whereas transitions to or from \(\mathsf O\) must satisfy the minimum up/down-time requirements. An event does not specify period-by-period releases or pumping quantities. Instead, the within-event LP in Section \ref{sec:event_lp_modules} determines the minimum-cost feasible dispatch over \(\{t,\ldots,j-1\}\) connecting the incoming and outgoing boundaries; if no such dispatch exists, \(c_{t,j}(s_t,e_t)=+\infty\).

The outgoing ramping boundary is linked to the terminal generation output of the within-event LP by
\begin{align}
\bar H_j =
\begin{cases}
H_{j-1}^{\mathsf O}, & \text{if } x_t\in\{\mathsf G,\mathsf{SC}\},\\
0, & \text{if } x_t\in\{\mathsf P,\mathsf O\},
\end{cases}
\end{align}
where $H_{j-1}^{\mathsf O}$ is the terminal turbine output of a generating or HSC event and is absent when $x_t\in\{\mathsf P, \mathsf O\}$.
If \(j\le T\), \(x_t\in\{\mathsf G, \mathsf{SC}\}\), and \(x^\dagger\in\{\mathsf P,\mathsf O\}\), event feasibility additionally requires $\bar H_j\le\overline V$.
This condition enforces the time-indexed ramp-down constraint when turbine output becomes zero at stage \(j\).
Finally, the counting state $\tau_t$ updates deterministically:
\begin{align}
\tau_j = \begin{cases} L-1, &x_t=\mathsf O,\ x^\dagger\in\mathcal X\setminus\{\mathsf O\},\\
l-1, &x_t\in\mathcal X\setminus\{\mathsf O\},\ x^\dagger=\mathsf O,\\
\max\{\tau_t-(j-t),0\}, &\text{otherwise}.
\end{cases}
\label{eq:tau}
\end{align}
Thus, executing event $e_t$ from $s_t$ produces the successor state $s_j=(x^\dagger,M_j,\bar H_j,\tau_j)$. Let $c_{t,j}(s_t,e_t)$ denote the optimal within-event operating cost.
The boundary transition cost is
\begin{align}
\Gamma_{t,j}(x_t,x^\dagger)=
\begin{cases}
SU_j, &j\le T,\ x_t=\mathsf O,\ x^\dagger\in\mathcal X\setminus\{\mathsf O\},\\
SD_j, &j\le T,\ x_t\in\mathcal X\setminus\{\mathsf O\},\ x^\dagger=\mathsf O,\\
0, &\text{otherwise}.
\end{cases}
\label{eq:gamma}
\end{align}
Transitions among \(\mathsf G\), \(\mathsf P\), and \(\mathsf{SC}\) do not incur fixed costs.
The event-based Bellman equation for $t\in\mathcal T$ is
\begin{equation}
\label{eq:DP}
V_t(s_t)=\min_{e_t\in\mathcal E(s_t)}
\Big\{c_{t,j}(s_t,e_t)+\Gamma_{t,j}(x_t,x^\dagger)+V_j(s_j)
\Big\},    
\end{equation}
with terminal condition $V_{T+1}(\cdot)\equiv 0$. 
As is standard in the literature \cite{zhai2010optimization,guan2018polynomial,Xiao2025, mu2026tackling, perez2010optimal, MartinRardinCampbell1990, deLimaEtAl2022}, the finite-horizon DP induces a network whose nodes are boundary states and whose arcs are events. 
Unlike conventional networks with discrete states or precomputed scalar arc costs, our state additionally includes continuous reservoir and ramping conditions required for within-event dispatch. 
Each arc embeds an LP that enforces reservoir balance, ramping, and operating constraints throughout the event, while shared boundary states link consecutive arcs so that the reservoir and ramping constraints also hold across event boundaries.
This structure provides the primal polyhedra used in the subsequent convexifications.

\subsection{Mode-Specific Within-Event Operating LPs}
\label{sec:event_lp_modules}
We now define \(c_{t,j}(s_t,e_t)\) as the optimal value of a mode-specific LP.
These costs are defined through epigraph variables \(\varphi_i\), whose lower bounds encode piecewise-linear physical costs, electricity payments or revenues, and water values. In the following LP blocks, \(\lambda_i\) denotes the electricity price at stage \(i\), and \(\nu\) denotes the marginal water value. The physical generation and pumping costs are represented by convex piecewise-linear functions. Specifically, \(a^{\mathsf G}_{i,m}\) and \(b^{\mathsf G}_{i,m}\), \(m=1,\ldots,N_{\mathsf G}\), define the generation-cost pieces, while \(a^{\mathsf P}_{i,m}\) and \(b^{\mathsf P}_{i,m}\), \(m=1,\ldots,N_{\mathsf P}\), define the pumping-cost pieces.\par
\ \\
\textbf{Generating block.} Consider the state $s_t = (\mathsf{G}, M_t,\bar{H}_t, \tau_t)$ of the generating block. The optimal operating cost of remaining in this mode from stage $t$ to $j-1$ is given by
\begin{equation}
c_{t,j}(s_t,e_t)=\min \sum_{i=t}^{j-1}\varphi_{i}^{\mathsf{G}}.
\label{eq:cG-primal-obj}
\end{equation}
The decision variables include generation outputs $H_i^{\mathsf{O}}$, epigraph variables $(\varphi_i^{\mathsf{G}},\phi_i^{\mathsf{G}})$ for $i\in\{t,\ldots,j-1\}$, and the reservoir levels $M_i$ for $i\in\{t,\ldots,j\}$. 
For any event $e_t=(j,x^{\dagger}, M_j, \bar H_j)$, let $M_{s_t}$ denote the reservoir component of $s_t$, and let $M_{e_t}$ and $\bar{H}_{e_t}$ denote the reservoir and ramping-boundary components of $e_t$, respectively.
%Denote the initial and terminal values as $:=M_t, M_{e_t}:=M_j$, and $\bar H_{e_t}:=\bar H_j$. 
Then, for $i\in\{t,\ldots,j-1\}$, we have
\begin{subequations}
\begin{align}
&\underline{B}_{i}^{\mathsf{O}}\leq H_{i}^{\mathsf{O}}\leq\bar{B}_{i}^{\mathsf{O}}, \label{eq:cG-primal-cap}\\
& -\overline V \leq H_{i}^{\mathsf{O}} - H_{i-1}^{\mathsf{O}} \leq \overline V,\label{eq:cG-primal-ramp-up}\\
%& H_{i-1}^{\mathsf{O}}-H_{i}^{\mathsf{O}}\leq\overline V,\label{eq:cG-primal-ramp-down}\\
& H_{t-1}^{\mathsf{O}}=\bar{H}_t, \label{eq:cG-primal-ramp-initial}\\
& M_{i+1} = M_{i} -\mu_{i}H_{i}^{\mathsf{O}}-\Lambda_{i} + R_{i}, \label{eq:cG-primal-reservoir}\\
& M_{t} = M_{s_t}, \label{eq:cG-primal-init-M}\\
& M_j = M_{e_t}, \label{eq:cG-primal-terminal-M}\\
& H_{j-1}^{\mathsf O}=\bar H_{e_t}, \label{eq:cG-primal-terminal-H}\\
& 0\leq M_{i}\leq \overline{M}, \label{eq:cG-primal-M-bounds}\\
& \varphi_{i}^{\mathsf{G}}\geq\phi^{\mathsf{G}}_{i}+ \nu\mu_i H_i^{\mathsf{O}} - \lambda_{i} H_{i}^{\mathsf{O}}, \label{eq:cG-primal-net-epi}\\
& \phi^{\mathsf{G}}_{i} \geq a^{\mathsf{G}}_{i,m} H_{i}^{\mathsf{O}} + b^{\mathsf{G}}_{i,m},\ m\in\{1, \dots, N_{\mathsf{G}}\}. \label{eq:cG-primal-PL}
\end{align}
\end{subequations}

Constraints \eqref{eq:cG-primal-cap}--\eqref{eq:cG-primal-ramp-up} impose generation capacity and ramping limits, with \eqref{eq:cG-primal-ramp-initial} linking the block to the previous generation level. \eqref{eq:cG-primal-reservoir} governs reservoir dynamics, and \eqref{eq:cG-primal-init-M}--\eqref{eq:cG-primal-M-bounds} enforce the incoming and outgoing reservoir boundaries, the outgoing ramping boundary, and the storage limits. \eqref{eq:cG-primal-net-epi}--\eqref{eq:cG-primal-PL} define the net operating cost using an epigraph formulation with generation cost functions that are convex and piecewise linear.\par
\ \\
\textbf{Pumping block.} For the pumping state $s_t=(\mathsf{P}, M_t, \bar{H}_t, \tau_t)$, the minimized total cost of remaining in pumping mode from stage $t$ to $j-1$ is given by
\begin{align}
c_{t,j}(s_t,e_t):=
\min & \sum_{i=t}^{j-1} \varphi_{i}^{\mathsf{P}}.
\end{align}
The decision variables include pumping power $H_i^{\mathsf{I}}$, epigraph variables $(\varphi_i^{\mathsf{P}},\phi_i^{\mathsf{P}})$ for $i\in\{t,\ldots,j-1\}$, and reservoir levels $M_i$ for $i\in\{t,\ldots,j\}$.
The spillage $\Lambda_i$ and natural inflow $R_i$ are known constants.
We have the following constraints for each $i\in\{t,\ldots,j-1\}$:
\begin{subequations}
\label{eq:cP-primal}
\begin{align}
&\underline{B}_{i}^{\mathsf{I}}\leq H_{i}^{\mathsf{I}} \leq\bar{B}_{i}^{\mathsf{I}}, \label{eq:cP-cap}\\
& M_{i+1}= M_{i}-\Lambda_{i} + R_{i} + \alpha_{i}H_{i}^{\mathsf{I}},\label{eq:cP-reservoir}\\
& M_{t} = M_{s_t}, \label{eq:cP-primal-init-M}\\
& M_j=M_{e_t}, \label{eq:cP-primal-terminal-M}\\
& \bar H_{e_t}=0, \label{eq:cP-primal-terminal-H}\\
& 0 \leq M_{i} \leq \overline{M},\label{eq:cP-M-bounds}\\
& \varphi_{i}^{\mathsf{P}}\geq\phi^{\mathsf{P}}_{i}
+ \lambda_{i} H_{i}^{\mathsf{I}}
- \nu \alpha_{i}H_{i}^{\mathsf{I}},\label{eq:cP-net-epi}\\
& \phi^{\mathsf{P}}_{i}\geq
a^{\mathsf{P}}_{i,m} H_{i}^{\mathsf{I}} + b^{\mathsf{P}}_{i,m},\ m\in\{1,\dots,N_{\mathsf{P}}\}.\label{eq:cP-PL}
\end{align}
\end{subequations}
Similar to the generating block, constraints~\eqref{eq:cP-cap}--\eqref{eq:cP-PL} define pumping capacity limits, reservoir dynamics constraint, the initial and terminal conditions and storage bounds, and the pumping-cost function that is convex and piecewise linear.

%Constraint~\eqref{eq:cP-cap} imposes pumping capacity limits, while \eqref{eq:cP-reservoir} governs reservoir dynamics in pumping mode. Constraints \eqref{eq:cP-primal-init-M}--\eqref{eq:cP-M-bounds} impose the initial and terminal conditions and storage bounds. Constraints \eqref{eq:cP-net-epi}--\eqref{eq:cP-PL} define the net operating cost via an epigraph formulation with a convex piecewise-linear pumping cost.\par
\ \\
\textbf{Hydraulic short-circuit block.}
For $s_t = (\mathsf{SC},M_t,\bar H_t,\tau_t)$, both turbine output and pumping input are permitted. 
The within-event operating cost is
\begin{align}
c_{t,j}(s_t,e_t) = \min\ \sum_{i=t}^{j-1}\varphi_i^{\mathsf{SC}}.
\label{eq:cSC-primal-obj}
\end{align}
The decision variables are $H_i^{\mathsf O}$, $H_i^{\mathsf I}$, $\varphi_{i}^{\mathsf{SC}}$, $\phi_{i}^{\mathsf G}$, and $\phi_i^{\mathsf P}$ for $i=t,\ldots, j-1$, together with $M_i$ for $i=t,\ldots, j$.
The HSC LP imposes the generation-capacity and turbine-ramping constraints
\eqref{eq:cG-primal-cap}--\eqref{eq:cG-primal-ramp-initial}, the boundary and reservoir-capacity constraints
\eqref{eq:cG-primal-init-M}--\eqref{eq:cG-primal-M-bounds}, the pumping-capacity constraint
\eqref{eq:cP-cap}, and the epigraph constraints
\eqref{eq:cG-primal-PL} and \eqref{eq:cP-PL} for the generation and pumping costs. The separate reservoir balance and net cost constraints of generating and pumping blocks are replaced, for \(i=t,\ldots,j-1\), by
\begin{subequations}
\label{eq:cSC-primal}
\begin{align}
M_{i+1}
&=
M_i-\mu_iH_i^{\mathsf O}-\Lambda_i+R_i
+\alpha_iH_i^{\mathsf I},
\label{eq:cSC-reservoir}\\
\varphi_i^{\mathsf{SC}}
&\ge
\phi_i^{\mathsf G}+\phi_i^{\mathsf P}
+(\nu\mu_i-\lambda_i)H_i^{\mathsf O}
+(\lambda_i-\nu\alpha_i)H_i^{\mathsf I}.
\label{eq:cSC-net-epi}
\end{align}
\end{subequations}
Hence, the HSC block combines the generating and pumping operating constraints.\\
\ \\
\textbf{Offline block.} For completeness, when $s_t=(\mathsf O,M_t,\bar H_t,\tau_t)$ and the plant remains offline over $\{t,\ldots,j-1\}$, we define $c_{t,j}(s_t,e_t)$ using the deterministic reservoir evolution with zero generation and pumping for $i=t,\ldots,j-1$,
\begin{align}
H_i^{\mathsf O}&=H_i^{\mathsf I}=0,\\
M_{i+1}&=M_i-\Lambda_i+R_i.
\end{align}
The event is feasible only if $M_i\in[0,\overline M]$ for $i=t+1,\ldots,j$ and $\bar H_j=0$. In that case, we set \(c_{t,j}(s_t,e_t)=0\).

The generating, pumping, HSC, and offline blocks differ only in their mode-specific physical and cost constraints. Each block is conditioned jointly on the incoming boundary \((M_t,\bar H_t)\) and outgoing boundary \((M_j,\bar H_j)\). 
Together with the common transition rules, these four blocks define the event feasibility and cost \(c_{t,j}(s_t,e_t)\) used in the Bellman equation. 
Consequently, the event-based formulation remains exactly equivalent to the time-indexed MILP.
\begin{proposition}
\label{prop:1}
Suppose the event set, state transition rules, and within-event LP blocks are defined as in Section \ref{sec:event_model}. Then the Bellman equation~\eqref{eq:DP} is equivalent to the time-indexed MILP in Section \ref{sec:MILP}. Every feasible solution of either formulation induces a feasible solution of the other with the same cost.
\end{proposition}

\subsection{A Finite-Dimensional Convexification of the
Continuous-State DP}
\label{sec:finite_dimensional_convexification}

The Bellman equation~\eqref{eq:DP} combines a finite set of discrete
event paths with continuous reservoir, ramping, and within-event
decisions. We enumerate only the discrete paths and retain all
continuous variables in a boundary-coupled polyhedron associated with
each path. Perspective scaling of these pathwise polyhedra then yields
an exact finite-dimensional convexification without state
discretization.

Recall that \(q\) collects the variables of the time-indexed formulation, let \(\mathcal Y\) denote its feasible set, and let \(C(q)\) denote its total cost. 
%Redundant unbounded cost-epigraph slacks are excluded from \(q\) and retained only as auxiliary variables representing \(C(q)\).
Recall that the discrete event state and action are $\tilde s_t=(t,x_t,\tau_t)$ and 
$\tilde a_t=(j,x^\dagger)$.
Let \(\widetilde\Omega\) contain all complete discrete event paths that satisfy~\eqref{eq:tau}, begin from the prescribed initial state, and terminate at stage \(T+1\). 
Since the horizon, mode set, and counter set are finite, \(\widetilde\Omega\) is finite, although potentially exponential in \(T\).

Fix \(\tilde\omega\in\widetilde\Omega\), let $1=t_1<t_2<\cdots<t_{K(\tilde\omega)+1}=T+1$ denote the corresponding event boundaries. Then event \(k\) operates in mode \(x_{t_k}\) over \(\{t_k,\ldots,t_{k+1}-1\}\).
For each \(i\in\mathcal T\), the path fixes \(y_i^{\mathsf G}(\tilde\omega)\), \(y_i^{\mathsf P}(\tilde\omega)\), and \(y_i^{\mathsf{SC}}(\tilde\omega)\), and hence also fixes \(y_i(\tilde\omega)\), \(g_i(\tilde\omega)\), \(p_i(\tilde\omega)\), \(u_i(\tilde\omega)\), and
\(d_i(\tilde\omega)\).
Set \(g_0(\tilde\omega)=0\), and introduce a path weight \(\rho_{\tilde\omega}\geq0\), with
\begin{equation}
\label{eq:FD-path-weight}
\sum_{\tilde\omega\in\widetilde\Omega} \rho_{\tilde\omega}=1.
\end{equation}

For every path, introduce the scaled variables $\widetilde M_i^{\tilde\omega}$, $\widetilde H_i^{\mathsf O,\tilde\omega}$, $\widetilde H_i^{\mathsf I,\tilde\omega}$, $\widetilde\phi_i^{\mathsf G,\tilde\omega}$, $\widetilde\phi_i^{\mathsf P,\tilde\omega}$ and $\widetilde h_i^{\tilde\omega}$, i.e., we define
\begin{equation}
\label{eq:FD-scaled-variable-definitions}
\begin{aligned}
\widetilde M_i^{\tilde\omega}
&=
\rho_{\tilde\omega}M_i^{\tilde\omega},
&
\widetilde H_i^{\mathsf O,\tilde\omega}
&=
\rho_{\tilde\omega}H_i^{\mathsf O,\tilde\omega},
&
\widetilde H_i^{\mathsf I,\tilde\omega}
&=
\rho_{\tilde\omega}H_i^{\mathsf I,\tilde\omega},
\\
\widetilde\phi_i^{\mathsf G,\tilde\omega}
&=
\rho_{\tilde\omega}\phi_i^{\mathsf G,\tilde\omega},
&
\widetilde\phi_i^{\mathsf P,\tilde\omega}
&=
\rho_{\tilde\omega}\phi_i^{\mathsf P,\tilde\omega},
&
\widetilde h_i^{\tilde\omega}
&=
\rho_{\tilde\omega}h_i^{\tilde\omega}.
\end{aligned}
\end{equation}
These quantities are defined directly as LP variables, without introducing bilinear products between path weights and unscaled variables. They satisfy
\begin{subequations}
\label{eq:FD-path-physical}
\begin{align}
\widetilde M_1^{\tilde\omega}
&=
M^{\mathrm{Initial}}\rho_{\tilde\omega},\\
0&\leq\widetilde M_i^{\tilde\omega}
\leq
\overline M\rho_{\tilde\omega},\\
g_i(\tilde\omega)\underline B_i^{\mathsf O}
\rho_{\tilde\omega}&\leq
\widetilde H_i^{\mathsf O,\tilde\omega}
\leq
g_i(\tilde\omega)\bar B_i^{\mathsf O}
\rho_{\tilde\omega},\\
p_i(\tilde\omega)\underline B_i^{\mathsf I}
\rho_{\tilde\omega}&\leq 
\widetilde H_i^{\mathsf I,\tilde\omega}
\leq
p_i(\tilde\omega)\bar B_i^{\mathsf I}
\rho_{\tilde\omega},\\
-\overline V g_{i-1}(\tilde\omega)\rho_{\tilde\omega}&\leq \widetilde H_i^{\mathsf O,\tilde\omega}
-\widetilde H_{i-1}^{\mathsf O,\tilde\omega}
\leq
\overline V g_i(\tilde\omega)\rho_{\tilde\omega},
\label{eq:FD-path-ramping}\\
\widetilde M_{i+1}^{\tilde\omega}
&=
\widetilde M_i^{\tilde\omega}
+(R_i-\Lambda_i)\rho_{\tilde\omega}
\notag\\
&\quad
-\mu_i\widetilde H_i^{\mathsf O,\tilde\omega}
+\alpha_i\widetilde H_i^{\mathsf I,\tilde\omega}.
\end{align}
\end{subequations}
The reservoir bounds hold for \(i=1,\ldots,T+1\), and the remaining constraints hold for \(i\in\mathcal T\). We set \(\widetilde H_0^{\mathsf O,\tilde\omega}=0\), consistently with the
initial offline condition.
Because a single variable \(\widetilde M_i^{\tilde\omega}\) is used at each physical stage, the terminal reservoir of one event is automatically the initial reservoir of the next. 
Likewise, \eqref{eq:FD-path-ramping} links consecutive turbine outputs directly; the event-boundary value \(\bar H_t\) is therefore implicit and need not be duplicated.

The epigraph of the scaled operating-cost function is
\begin{subequations}
\label{eq:FD-path-cost-epigraph}
\begin{align}
\widetilde\phi_i^{\mathsf G,\tilde\omega}
&\geq
a_{i,m}^{\mathsf G}
\widetilde H_i^{\mathsf O,\tilde\omega}+b_{i,m}^{\mathsf G}
g_i(\tilde\omega)\rho_{\tilde\omega},
\label{eq:FD-path-G-cost}\\
\widetilde\phi_i^{\mathsf P,\tilde\omega}
&\geq
a_{i,m}^{\mathsf P}
\widetilde H_i^{\mathsf I,\tilde\omega}+b_{i,m}^{\mathsf P}
p_i(\tilde\omega)\rho_{\tilde\omega},
\label{eq:FD-path-P-cost}\\
\widetilde h_i^{\tilde\omega}
&\geq
\widetilde\phi_i^{\mathsf G,\tilde\omega}
+\widetilde\phi_i^{\mathsf P,\tilde\omega}
+(\nu\mu_i-\lambda_i)
\widetilde H_i^{\mathsf O,\tilde\omega}
\notag\\
&\quad+(\lambda_i-\nu\alpha_i)
\widetilde H_i^{\mathsf I,\tilde\omega}.
\label{eq:FD-path-net-cost}
\end{align}
\end{subequations}
These inequalities are imposed for every \(i=1,\ldots,T\), with the first and second also imposed for \(m=1,\ldots,N_{\mathsf G}\) and \(m=1,\ldots,N_{\mathsf P}\), respectively.
If the path is in mode \(\mathsf G\), only turbine output is active; if it is in mode \(\mathsf P\), only pumping input is active; in mode \(\mathsf{SC}\), \(g_i(\tilde\omega)=p_i(\tilde\omega)=1\), so both physical costs and both energy terms are present; and in mode \(\mathsf O\), all power variables and the minimum feasible operating cost are zero. 
Thus \eqref{eq:FD-path-physical}--\eqref{eq:FD-path-cost-epigraph} are exactly the perspective-scaled forms of the four within-event blocks.

The scaled cost associated with path \(\tilde\omega\) is
\begin{align}
\widetilde C^{\tilde\omega}
:={}&
\sum_{i=1}^{T}\widetilde h_i^{\tilde\omega}
+
\rho_{\tilde\omega}
\sum_{k=1}^{K(\tilde\omega)}
\Gamma_{t_k,t_{k+1}}
\left(x_{t_k},x_{t_{k+1}}\right).
\label{eq:FD-path-total-cost}
\end{align}
The first term represents within-event operating costs, while the second collects the transition costs along the discrete path.

The time-indexed continuous variables are reconstructed by
\begin{subequations}
\label{eq:FD-reconstruction}
\begin{align}
\bar M_i
&=
\sum_{\tilde\omega\in\widetilde\Omega}
\widetilde M_i^{\tilde\omega},\quad i=1,\ldots,T+1,
\label{eq:FD-reconstruct-M}\\
\bar H_i^{\mathsf O}
&=
\sum_{\tilde\omega\in\widetilde\Omega}
\widetilde H_i^{\mathsf O,\tilde\omega},
\quad 
\bar H_i^{\mathsf I}
=
\sum_{\tilde\omega\in\widetilde\Omega}
\widetilde H_i^{\mathsf I,\tilde\omega},
\quad i=1,\ldots,T.
\label{eq:FD-reconstruct-H}
\end{align}
\end{subequations}
For each
\(v\in\{y^{\mathsf G},y^{\mathsf P},y^{\mathsf{SC}},y,g,p,u,d\}\),
the corresponding discrete or transition variable is reconstructed by
\begin{align}
\bar v_i
&=
\sum_{\tilde\omega\in\widetilde\Omega}
\rho_{\tilde\omega}v_i(\tilde\omega),
\quad i=1,\ldots,T.
\label{eq:FD-reconstruct-discrete}
\end{align}
Let \(\bar q\) collect these reconstructed variables and define
\begin{align}
\bar C
=
\sum_{\tilde\omega\in\widetilde\Omega}
\widetilde C^{\tilde\omega}.
\label{eq:FD-total-cost}
\end{align}
Let \(\mathcal P^{\mathsf{FD}}\) denote the feasible region defined by \eqref{eq:FD-path-weight}--\eqref{eq:FD-total-cost}. 
The resulting formulation is a finite-dimensional LP because both the horizon and \(\widetilde\Omega\) are finite.

For every path with \(\rho_{\tilde\omega}>0\), dividing its scaled variables by \(\rho_{\tilde\omega}\) recovers a feasible schedule in \(\mathcal Y\) whose discrete decisions follow \(\tilde\omega\). 
If \(\rho_{\tilde\omega}=0\), the scaled physical variables vanish, while the auxiliary epigraph variables are excluded from \(\bar q\) and may be set to zero without loss. 
Hence, every reconstructed point is a convex combination of schedules in \(\mathcal Y\).

Conversely, \textbf{Proposition~\ref{prop:1}} associates every schedule in \(\mathcal Y\) with a complete discrete event path. 
Grouping any convex combination of feasible schedules by their induced paths and summing
their weighted continuous variables produces a feasible point of \(\mathcal P^{\mathsf{FD}}\). These two constructions establish the convex-hull equality below.

\begin{proposition}[Finite-dimensional continuous-state convexification]
\label{prop:finite_dimensional_convexification}
Suppose the event-state convention, transition rules, and within-event feasible sets are consistent with the time-indexed formulation. Then
\[
\operatorname{proj}_{\bar q}
\left(\mathcal P^{\mathsf{FD}}\right)
=
\operatorname{conv}(\mathcal Y),
\]
and
\[
\min_{(\bar q,\cdot)\in\mathcal P^{\mathsf{FD}}}
\bar C
=
V_1(s_1)
=
\min_{q\in\mathcal Y}C(q).
\]
Moreover, an optimal solution can be chosen with unit weight assigned to a single optimal discrete event path.
\end{proposition}
\textbf{Proposition~\ref{prop:finite_dimensional_convexification}} shows that all nonconvexities in the unit-level PSH problem can be attributed to the selection of a discrete event path. 
Its exact convex hull is therefore obtained by convexifying only this selection, i.e., a finite union of boundary-coupled polyhedra indexed by discrete event paths, while leaving the reservoir, ramping, and operating variables continuous within each path. The proposition thus yields an exact convex-hull extended formulation whose projection onto the original time-indexed scheduling variables equals the convex hull of all feasible schedules, rather than merely matching the optimal value for a given price trajectory.
%Conditional on such a path, the continuous reservoir, ramping, and operating decisions form a boundary-coupled polyhedron. 
%Thus, its exact convex hull is obtained by convexifying the selection among these paths, while retaining the continuous reservoir, ramping, and operating decisions within each path.
%The proposition therefore provides an ideal extended formulation in the original scheduling-variable space, rather than merely an alternative formulation that attains
%the same optimal value for a particular price trajectory. 
%In particular, its projection recovers the entire convex hull of feasible time-indexed schedules while keeping the reservoir and ramping states continuous. 
%This objective-independent characterization makes the formulation a natural unit-level building block for decomposed system scheduling and market-clearing models.

The exact convex-hull formulation may nevertheless contain exponentially many complete discrete paths. 
This motivates two complementary uses of the same structure.

\section{Finite-Grid Event Network and LP Approximation}
\label{sec:LP}
We first restrict the reservoir and ramping boundary states to finite grids, allowing paths that reach the same boundary state to share nodes and arcs in a finite event network. 
The resulting formulation applies the same perspective scaling as the pathwise convexification, now at the arc level.

We discretize the continuous reservoir range $[0,\overline{M}]$ and the ramping-boundary range $[0,\overline{H}^{\mathsf O}]$ onto finite points $\mathcal M:=\{m^1,\ldots,m^{N_M}\}$ and $\mathcal H:=\{h^1,\ldots,h^{N_H}\}$, respectively, where $0\in\mathcal H$ and $M^{\text{Initial}}\in\mathcal M$.
A discretized state is therefore $\hat{s}_t=(x_t,M_t,\bar{H}_t,\tau_t)\in\hat{\mathcal{S}}_t$, where $\hat{\mathcal{S}}_t:=\mathcal X\times \mathcal M\times \mathcal H\times \{0,\ldots,\tau_{\max}\}$.
We take the initial node to be $\hat s_1=(\mathsf O,\,M^{\text{Initial}},\,0,\,0)$.

Each admissible action \(\hat a_t=(j,x^\dagger,M_j,\bar H_j)\) defines an arc from \(\hat s_t=(x_t,M_t,\bar H_t,\tau_t)\) to $\hat s_j=(x^\dagger,M_j,\bar H_j,\tau_j)$, where \(\tau_j\) follows~\eqref{eq:tau}. 
The arc is retained only if the transition is admissible and the corresponding within-event LP is
feasible. Let \(\hat{\mathcal{A}}(\hat{s}_t)\subseteq\{t+1,\ldots,T+1\}\times\mathcal X\times\mathcal{M}\times\mathcal{H}\) denote the resulting arc set.

To select a feasible path, we associate each arc $\hat a_t\in \hat{\mathcal A}(\hat s_t)$ with a nonnegative flow variable $0\leq \pi_{\hat s_t,\hat a_t}\leq 1$.
Sending one unit of flow from the initial node and imposing flow conservation at every intermediate node therefore yields a source-to-sink unit-flow model on an event network. 

To obtain a single LP without precomputing arc costs, we associate a copy of the within-event operating LP with each arc using the same perspective scaling as in Section \ref{sec:finite_dimensional_convexification}.
Consider a node \(\hat s_t=(x_t,M_t,\bar H_t,\tau_t)\) and an arc \(\hat a_t=(j,x^\dagger,M_j,\bar H_j)\).
Because the unit remains in mode \(x_t\) over \(\{t,\ldots,j-1\}\), the activation variables $(g_i,p_i)$ are fixed on this arc.

For an arc \((\hat s_t,\hat a_t)\), we introduce the corresponding scaled
operating variables $\tilde H_i^{\mathsf O}$, $\tilde H_i^{\mathsf I}$, $\tilde M_i$, $\tilde\phi_i^{\mathsf G}$, $\tilde\phi_i^{\mathsf P}$, and $\tilde h_i$ directly as LP variables and apply the same scaling used in \eqref{eq:FD-scaled-variable-definitions}, with the path weight \(\rho_{\tilde\omega}\) replaced by \(\pi_{\hat s_t,\hat a_t}\). 
No bilinear equalities between \(\pi_{\hat s_t,\hat a_t}\) and unscaled operating variables are imposed. The common arcwise operating block is
\begin{subequations}
\label{eq:disc-common-block}
\begin{align}
\tilde c_t(\hat s_t,\hat a_t)
&=
\sum_{i=t}^{j-1}\tilde h_i,
\label{eq:disc-common-cost}\\
g_i\underline B_i^{\mathsf O}
\pi_{\hat s_t,\hat a_t}
&\leq
\tilde H_i^{\mathsf O}
\leq
g_i\bar B_i^{\mathsf O}
\pi_{\hat s_t,\hat a_t},
\label{eq:disc-common-G-cap}\\
p_i\underline B_i^{\mathsf I}
\pi_{\hat s_t,\hat a_t}
&\leq
\tilde H_i^{\mathsf I}
\leq
p_i\bar B_i^{\mathsf I}
\pi_{\hat s_t,\hat a_t},
\label{eq:disc-common-P-cap}\\
-\overline V\pi_{\hat s_t,\hat a_t}
&\leq
\tilde H_i^{\mathsf O}
-\tilde H_{i-1}^{\mathsf O}
\leq
\overline V\pi_{\hat s_t,\hat a_t},
\label{eq:disc-common-ramping}\\
\tilde H_{t-1}^{\mathsf O}
&=
\bar H_t\pi_{\hat s_t,\hat a_t},\ \ 
%\label{eq:disc-common-H-in}\\
\tilde H_{j-1}^{\mathsf O}
=
\bar H_j\pi_{\hat s_t,\hat a_t}.
\label{eq:disc-common-H-boundary}\\
\tilde M_t
&=
M_t\pi_{\hat s_t,\hat a_t},\ \ 
\tilde M_j
=
M_j\pi_{\hat s_t,\hat a_t},
\label{eq:disc-common-M-boundary}\\
0
\leq
\tilde M_i
&\leq
\overline M\pi_{\hat s_t,\hat a_t},
\label{eq:disc-common-M-bounds}\\
\tilde M_{i+1}
&=
\tilde M_i
+(R_i-\Lambda_i)\pi_{\hat s_t,\hat a_t}
\notag\\
&\quad
-\mu_i\tilde H_i^{\mathsf O}
+\alpha_i\tilde H_i^{\mathsf I}, \label{eq:disc-common-reservoir}
\end{align}
\end{subequations}
with the scaled cost epigraphs defined as
\begin{align}
\tilde\phi_i^{\mathsf G}
&\geq
a_{i,m}^{\mathsf G}\tilde H_i^{\mathsf O}
+
b_{i,m}^{\mathsf G}g_i\pi_{\hat s_t,\hat a_t},
\quad 
m=1,\ldots,N_{\mathsf G},
\label{eq:disc-common-G-cost}\\
\tilde\phi_i^{\mathsf P}
&\geq
a_{i,m}^{\mathsf P}\tilde H_i^{\mathsf I}
+
b_{i,m}^{\mathsf P}p_i\pi_{\hat s_t,\hat a_t},
\quad 
m=1,\ldots,N_{\mathsf P},
\label{eq:disc-common-P-cost}\\
\tilde h_i
&\geq
\tilde\phi_i^{\mathsf G}
+\tilde\phi_i^{\mathsf P}
+(\nu\mu_i-\lambda_i)\tilde H_i^{\mathsf O}
+(\lambda_i-\nu\alpha_i)\tilde H_i^{\mathsf I},
\label{eq:disc-common-net-cost}
\end{align}
The stagewise constraints hold for \(i=t,\ldots,j-1\), while
\eqref{eq:disc-common-M-bounds} holds for \(i=t,\ldots,j\).

When \(x_t=\mathsf G\), the pumping variables are fixed to zero and \eqref{eq:disc-common-block} reduces to the scaled generating block. 
When \(x_t=\mathsf P\), the turbine-output variables are fixed to zero and it reduces to the scaled pumping block. 
Both components are active when \(x_t=\mathsf{SC}\), whereas both are zero when \(x_t=\mathsf O\), leaving only the deterministic reservoir evolution and zero operating cost. 
Thus, the same constraints represent all four mode-specific within-event LPs.

For pumping and offline arcs, \eqref{eq:disc-common-G-cap} and \eqref{eq:disc-common-H-boundary} imply \(\bar H_j=0\) whenever the arc carries positive flow. 
For generating and HSC arcs, \eqref{eq:disc-common-H-boundary} passes the terminal turbine output to the successor state. 
If the successor mode is pumping or offline, the ramping constraint of the successor event
enforces the required transition to zero.

Let \(z\) collect all arc-flow variables and the scaled operating variables. Combining \eqref{eq:disc-common-block}--\eqref{eq:disc-common-net-cost} with the unit-flow constraints yields the finite-grid LP
\begin{subequations}
\label{eq:ALLINONE}
\begin{align}
\min_{z}\quad
& \sum_{t=1}^{T}\ \sum_{\hat{s}_t}\ \sum_{\hat{a}_t}
\Big[\tilde c_{t}(\hat{s}_t,\hat{a}_t)+\Gamma(\hat s_t,\hat a_t)\,\pi_{\hat s_t,\hat a_t}\Big], \label{eq:ALLINONE_obj}\\
\text{s.t.}\quad
& \sum_{\hat a_1\in\hat{\mathcal A}(\hat s_1)} \pi_{\hat s_1,\hat a_1} = 1, \label{eq:ALLINONE_source} \\
& \sum_{\hat a_t\in\hat{\mathcal A}(\hat s_t)} \pi_{\hat s_t, \hat a_t}
= \sum_{t',\hat s_{t'},\hat a_{t'}}
\pi_{\hat s_{t'},\hat a_{t'}}\,\mathbbm{1}\!\left\{\hat f(\hat s_{t'},\hat a_{t'})=\hat s_t\right\},\notag\\
&\hspace{2cm} \forall t\in\{2,\ldots,T\},\ \forall \hat s_t\in\hat{\mathcal S}_t, \label{eq:ALLINONE_flow}\\
& \pi_{\hat s_t,\hat a_t}\ge 0,\quad
\forall t\in\mathcal{T},\ \forall \hat{s}_t\in\hat{\mathcal S}_t,\ \forall \hat{a}_t\in\hat{\mathcal A}(\hat{s}_t). \label{eq:ALLINONE_nonneg}
\end{align}
\end{subequations}
Here, \(\tilde c_t(\hat s_t,\hat a_t)\) is the scaled operating cost for the current mode, \(\Gamma(\hat s_t,\hat a_t)\) is the transition cost in \eqref{eq:gamma}, \(\hat f(\hat s_t,\hat a_t)\) denotes the successor node, and \(\mathbbm{1}(\cdot)\) is the indicator function. Since the formulation combines a unit-flow network with arc-separable operating models, it admits an integral optimum corresponding to a source-to-sink path, as formalized below.
\begin{proposition}
\label{prop:3}
The LP defined by \eqref{eq:disc-common-block}--\eqref{eq:ALLINONE} is an extended convex-hull formulation of the finite-grid event-network model. Hence, it admits an optimal solution with integral arc flows, and its optimal value equals that of the finite-grid DP.
\end{proposition}
In a conventional scalar-cost DP,
each within-arc operating problem is solved in advance and represented
in the network only by its optimal cost. The resulting network flow
formulation therefore describes path selection but does not retain the
underlying operating decisions. In contrast,
\textbf{Proposition~\ref{prop:3}} embeds each within-event LP directly
on its arc, retaining the associated dispatch, reservoir, and ramping
variables. It therefore characterizes the convex hull jointly in the
arc flows and physical operating variables, yielding an extended
formulation of the complete grid-restricted scheduling model.
%Unlike conventional convex-hull formulations of finite-state DPs, which characterize only path flows after scalar arc costs are specified, \textbf{Proposition~\ref{prop:3}} characterizes the convex hull jointly in the arc-flow and operating variables. Embedding each arc's within-event LP preserves its feasible operating trajectories and yields an extended formulation of the complete grid-restricted scheduling model.

%embeds the boundary-conditioned within-event LPs and convexifies the complete grid-restricted operating model in its arc-flow, dispatch, reservoir, and ramping variables, i.e., a full variable convex hull. 
%This objective-independent unit-level formulation can be embedded directly in network-constrained scheduling, used as an LP pricing subproblem in column generation, or represented through primal--dual optimality conditions in bilevel market models.

\textbf{Remark 1 (Comparison with stagewise discretization).}
\textit{The proposed formulation discretizes the reservoir and ramping states only at event boundaries, whereas conventional stagewise discretization restricts them at every stage. The stagewise DP model is recovered by restricting each event to a single stage. Hence, when the same grids are used, every stagewise-grid feasible schedule remains feasible in the event-boundary model, and the optimal cost of the event-based LP cannot exceed that of the stagewise DP. Table~\ref{tab:grid-resolution} shows that boundary-only discretization can substantially improve accuracy.}

\textbf{Proposition~\ref{prop:3}} also indicates that the LP introduces no relaxation error relative to the grid-restricted event model.
Its only loss relative to the continuous-state formulation is boundary-state discretization, which \textbf{Proposition~\ref{prop:4}} quantifies.
\begin{proposition}
\label{prop:4}
Let \(\Delta_M\) and \(\Delta_H\) denote the mesh sizes of the reservoir and ramping-boundary grids. Suppose every feasible continuous event admits a feasible grid perturbation with the same switching time and successor mode and boundary error at most \(C_B(\Delta_M+\Delta_H)\), where \(C_B\) is independent of the grid resolution. Then the finite-grid and continuous-state DP optimal values differ by at most \(C(\Delta_M+\Delta_H)\), where \(C\) is independent of the grid resolution.
\end{proposition}
%For standalone computation, the within-event LPs can be optimized independently to obtain scalar arc costs, after which the remaining problem is a shortest-path problem on the event DAG. 
%This specialized implementation is equivalent to the monolithic finite-grid LP, while the latter exposes the physical operating variables needed for decomposition and system-level embedding.
%At the system level, the finite-grid convex-hull LP can be embedded directly in network-constrained scheduling, used as an LP pricing subproblem in column generation, or represented through primal--dual optimality conditions in bilevel market models. 
%Because it exactly convexifies the grid-restricted unit model, these applications introduce no additional unit-level relaxation error beyond boundary-state discretization.

\textbf{Remark 2 (Computational Complexity).} \textit{Define \(\bar\tau:=\tau_{\max}+1\), \(N_M:=|\mathcal M|\), and \(N_H:=|\mathcal H|\). 
With four operating modes, the discretized state space at each stage contains \(4N_MN_H\bar\tau\) states. Consequently, the full event network contains \(O(TN_MN_H\bar\tau)\) states and, in the worst case, \(O(T^2N_M^2N_H^2\bar\tau)\) arcs. 
Once the arc costs have been computed, the remaining optimization reduces to a shortest-path problem and can be solved in \(O(|\hat{\mathcal S}|+|\hat{\mathcal A}|)\).}

For fixed inflow and spillage profiles, the event-network topology is independent of electricity prices and can therefore be reused across price scenarios. 
Varying inflows may change arc feasibility but can be represented using a common superset of arcs with scenario-specific constraints. 
The arc-embedded LP in \textbf{Proposition~\ref{prop:3}} also suggests a potential LP-recourse formulation with dual-based Benders cuts after parameterizing the boundary states and imposing nonanticipativity. 
Developing this stochastic extension is left for future research.

\section{Event-Based B\&B}
\label{sec:ebb}

This section presents an exact B\&B method for independently verifying the continuous-state event-based DP in Section \ref{sec:event_blocks}.
The method branches only on the discrete event decisions, while the reservoir and ramping boundary states remain continuous.

At stage \(t\), define the reduced state $\tilde s_t:=(t,x_t,\tau_t)$, and let \(\widetilde{\mathcal S}_t\) denote the corresponding finite state set. 
A B\&B node \(n\) consists of a fixed discrete partial path \(P_n\) ending at \(\tilde s_t\). 
Thus, distinct partial paths ending at the same reduced state are treated as distinct B\&B nodes.

\textbf{Branching.}
From \(\tilde s_t\), the algorithm enumerates feasible actions $\tilde a_t=(j,x^\dagger)\in\widetilde{\mathcal A}(\tilde s_t)$.
Under this action, the unit remains in mode \(x_t\) over \(\{t,\ldots,j-1\}\) and enters mode \(x^\dagger\) at stage \(j\).
The successor reduced state is $\tilde f(\tilde s_t,\tilde a_t) = (j,x^\dagger,\tau_j)$, where \(\tau_j\) is determined by \eqref{eq:tau}. 
Each child node appends one such action to \(P_n\). Only the discrete prefix is fixed; all continuous decisions on the prefix and the remaining horizon are jointly reoptimized at every child node.

\textbf{Upper bound.}
Once a complete discrete path is obtained, we solve its continuous path LP, which jointly optimizes the reservoir, ramping, and within-event operating decisions along that path. 
If the path LP is feasible and its objective value is smaller than the current incumbent, that value becomes the new upper bound \(\mathrm{UB}\).

\textbf{Lower bound.}
At node \(n\), we relax the remaining discrete event path into a full-horizon network flow. Arc-selection variables are allowed to be fractional, and products between these flows and the continuous boundary states are replaced by their McCormick envelopes. 
Because every feasible completion of node \(n\) remains feasible in this LP, its optimal value
provides a lower bound.

For each reduced state \(\tilde s_r\) and action \(\tilde a_r\in\widetilde{\mathcal A}(\tilde s_r)\), let \(\pi_{\tilde s_r,\tilde a_r}\in[0,1]\) denote the flow on the corresponding arc. The flows satisfy
\begin{subequations}
\label{eq:ebb-flow}
\begin{align}
\sum_{\tilde a_1}
\pi_{\tilde s_1,\tilde a_1}
&=1,
\label{eq:ebb-flow-source}\\
\sum_{\tilde a_r}
\pi_{\tilde s_r,\tilde a_r}
&=
\sum_{r'=1}^{r-1}
\sum_{\tilde s_{r'}, \tilde a_{r'}}
\pi_{\tilde s_{r'},\tilde a_{r'}}
\mathbbm 1
\left\{
\tilde f(\tilde s_{r'},\tilde a_{r'})
=\tilde s_r
\right\},
\label{eq:ebb-flow-conservation}
\end{align}
\end{subequations}
for \(r=2,\ldots,T\) and \(\tilde s_r\in\widetilde{\mathcal S}_r\).
The prefix defining node \(n\) is enforced by $\pi_{\tilde s_r,\tilde a_r}=1$ with $(\tilde s_r,\tilde a_r)\in P_n$.

Unlike the finite-grid formulation, the reservoir and ramping states associated with each reduced state remain continuous, with $0\leq M_{\tilde s_r}\leq\overline M$ and $0\leq\bar H_{\tilde s_r}\leq\overline H^{\mathsf O}$.
Consider an arc \((\tilde s_r,\tilde a_r)\) whose successor is \(\tilde s_j=\tilde f(\tilde s_r,\tilde a_r)\). 
We introduce \(m^{\mathrm{in}}_{\tilde s_r,\tilde a_r}\) and \(m^{\mathrm{out}}_{\tilde s_r,\tilde a_r}\) to represent the flow-weighted reservoir states at \(\tilde s_r\) and \(\tilde s_j\), respectively. 
The variables \(h^{\mathrm{in}}_{\tilde s_r,\tilde a_r}\) and \(h^{\mathrm{out}}_{\tilde s_r,\tilde a_r}\) represent the analogous flow-weighted ramping states.

The boundary state carried by an arc must depend on whether that arc is selected. In an integral event path, the exact relationships are $m^{\mathrm{in}}_{\tilde s_r,\tilde a_r} = \pi_{\tilde s_r,\tilde a_r}M_{\tilde s_r}$, $m^{\mathrm{out}}_{\tilde s_r,\tilde a_r}= \pi_{\tilde s_r,\tilde a_r}M_{\tilde s_j}$, $h^{\mathrm{in}}_{\tilde s_r,\tilde a_r}=\pi_{\tilde s_r,\tilde a_r}\bar H_{\tilde s_r}$, and $h^{\mathrm{out}}_{\tilde s_r,\tilde a_r}=\pi_{\tilde s_r,\tilde a_r}\bar H_{\tilde s_j}$.
Thus, a selected arc carries the boundary values of its incident states, whereas an unselected arc carries zero. 
Because the lower-bound formulation allows \(\pi_{\tilde s_r,\tilde a_r}\in[0,1]\), these identities become bilinear. For each intended product \(z=\pi q\), where
\(0\leq\pi\leq1\) and \(0\leq q\leq\overline q\), we impose McCormick envelope
\begin{equation}
0\leq z\leq\overline q\,\pi,
\qquad
z\leq q,
\qquad
z\geq q-\overline q(1-\pi).
\label{eq:ebb-mccormick}
\end{equation}
The lifted boundary values are propagated through the event network.
For every intermediate reduced state \(\tilde s_r\), we impose
\begin{subequations}
\label{eq:ebb-lifted-conservation}
\begin{align}
\sum_{\tilde a_r}
m^{\mathrm{in}}_{\tilde s_r,\tilde a_r}
&=
\sum_{r'=1}^{r-1}
\sum_{\tilde s_{r'},\tilde a_{r'}}
m^{\mathrm{out}}_{\tilde s_{r'},\tilde a_{r'}}
\mathbbm 1
\{
\tilde f(\tilde s_{r'},\tilde a_{r'})
=\tilde s_r
\},
\label{eq:ebb-reservoir-conservation}\\
\sum_{\tilde a_r}
h^{\mathrm{in}}_{\tilde s_r,\tilde a_r}
&=
\sum_{r'=1}^{r-1}
\sum_{\tilde s_{r'},\tilde a_{r'}}
h^{\mathrm{out}}_{\tilde s_{r'},\tilde a_{r'}}
\mathbbm 1
\{
\tilde f(\tilde s_{r'},\tilde a_{r'})
=\tilde s_r
\}.
\label{eq:ebb-ramping-conservation}
\end{align}
\end{subequations}
At the initial state, the boundary variables and their lifted counterparts satisfy $M_{\tilde s_1} =M^{\mathrm{Initial}}$, $\bar H_{\tilde s_1}=0$, $m^{\mathrm{in}}_{\tilde s_1,\tilde a_1}= M^{\mathrm{Initial}} \pi_{\tilde s_1,\tilde a_1}$, and $h^{\mathrm{in}}_{\tilde s_1,\tilde a_1}=0$, where $\tilde a_1\in\widetilde{\mathcal A}(\tilde s_1)$.

For each arc, we impose the mode-specific perspective block obtained from \eqref{eq:FD-path-physical}--\eqref{eq:FD-path-cost-epigraph}, replacing the path weight by the arc flow and the fixed grid boundaries by the lifted continuous boundary variables $\tilde M_r = m^{\mathrm{in}}_{\tilde s_r,\tilde a_r}$ and $\tilde M_j =m^{\mathrm{out}}_{\tilde s_r,\tilde a_r}$.
For generating and HSC arcs, we additionally impose $\tilde H_{r-1}^{\mathsf O} = h^{\mathrm{in}}_{\tilde s_r,\tilde a_r}$ and $\tilde H_{j-1}^{\mathsf O} =
h^{\mathrm{out}}_{\tilde s_r,\tilde a_r}$.
For pumping and offline arcs, \(h^{\mathrm{out}}_{\tilde s_r,\tilde a_r}=0\). 
If \(j\leq T\), \(x_r\in\{\mathsf G,\mathsf{SC}\}\), and \(x^\dagger\in\{\mathsf P,\mathsf O\}\), the ramp-down transition is enforced by $h^{\mathrm{out}}_{\tilde s_r,\tilde a_r}\leq \overline V\pi_{\tilde s_r,\tilde a_r}$.
The remaining mode-specific constraints follow the same arcwise specialization and are not repeated.
%All other operating constraints retain the perspective-scaled forms given in \S\ref{sec:LP}; thus, they need not be repeated here.

For \(\tilde a_r=(j,x^\dagger)\), define $\Gamma(\tilde s_r,\tilde a_r):=\Gamma_{r,j}(x_r,x^\dagger)$. Let \(\underline V(n)\) denote the optimal value of the resulting node relaxation:
\begin{equation}
\underline V(n)
=
\min
\sum_{r=1}^{T}
\sum_{\tilde s_r, \tilde a_r}
\left[
\tilde c_r(\tilde s_r,\tilde a_r)
+
\Gamma(\tilde s_r,\tilde a_r)
\pi_{\tilde s_r,\tilde a_r}
\right].
\label{eq:ebb-node-lb}
\end{equation}
Any feasible completion consistent with \(P_n\) induces an integral flow satisfying this relaxation. 
Therefore, \(\mathrm{LB}(n):=\underline V(n)\) is a valid lower bound on every completion of node \(n\). 
The node is pruned whenever \(\mathrm{LB}(n)\geq\mathrm{UB}\), up to the prescribed numerical tolerance. 
The following proposition formalizes finite termination and exactness.
\begin{proposition}
\label{prop:5}
The B\&B algorithm terminates after finitely many branching steps and returns the optimal value of the continuous-state event-based formulation.
\end{proposition}

\section{Numerical Results}
\label{sec:num}

We compare the finite-grid LP with the time-indexed MILP, using the event-based
B\&B only to verify theoretical exactness. The experiments include a \(T=24\) standalone
single-unit instance adapted from \cite{Qu2023} and modified IEEE RTS-24 cases with three PSH units over \(T=336\) and nine units over \(T=672\) (see, e.g., \cite{Zimmerman2011MATPOWER}). The system-level implementations use
the same column generation (CG) framework. Appendix~\ref{app:numerical-settings} provides the complete settings.
%All algorithms are implemented in C++, with the MILPs and LPs solved using Gurobi 12.0. The experiments are conducted on a workstation with a 24-core Intel(R) Core(TM) i9-14900K processor and 128~GB RAM.

\subsection{Results}
\label{sec:results}

Table~\ref{tab:exactness} shows that the finite-grid DP, LP, and event-based B\&B attain the same optimum of \(-92{,}360\) on common boundary-state grids, consistent with \textbf{Proposition~\ref{prop:3}}. 
%The grid-restricted event-based B\&B method also attains \(-92{,}360\), providing an independent confirmation of this optimum. 
In the last two rows, the boundary states are continuous. The continuous-state event-based B\&B method and the time-indexed MILP both attain \(-94{,}385\), illustrating the combined implications of \textbf{Propositions~\ref{prop:1}} and~\textbf{\ref{prop:5}}. 
Because rows 3 and 4 differ only in whether the boundary states are discretized, their \(2.15\%\) objective difference isolates the effect of discretization. Thus, the finite-grid LP introduces no additional relaxation error, consistent with \textbf{Proposition~\ref{prop:3}}.

\begin{table}[!h]
\centering
\caption{Exactness comparisons with a \(2\%\) reservoir-state grid}
\label{tab:exactness}
\renewcommand{\arraystretch}{1.15}
\begin{tabular}{ccc}
\hline
Boundary-state treatment & Method & Objective value (\$) \\ \hline
Grid-restricted & DP               & \(-92{,}360\) \\
Grid-restricted & LP               & \(-92{,}360\) \\
Grid-restricted & Event-based B\&B & \(-92{,}360\) \\
Continuous-state & Event-based B\&B & \(-94{,}385\) \\
Continuous-state & MILP             & \(-94{,}385\) \\ \hline
\end{tabular}
\end{table}

\begin{table*}[!b]
\centering
\scriptsize
\setlength{\dbltextfloatsep}{6pt plus 2pt minus 2pt}
\renewcommand{\arraystretch}{1.05}
\caption{Comparison of event-boundary and stagewise discretization
under reservoir-grid refinement at \(T=24\). The event-boundary LP uses
\(N_H=6\), whereas the stagewise DP uses a finer ramping grid with
\(N_H=30\).}
\label{tab:grid-resolution}
\begin{tabular}{rrrrrrrr}
\toprule
& & \multicolumn{3}{c}{Event-boundary LP (\(N_H=6\))}
& \multicolumn{3}{c}{Stagewise DP (\(N_H=30\))}\\
\cmidrule(lr){3-5}\cmidrule(lr){6-8}
\(\Delta_M/\overline M\) (\%)
& \(N_M\)
& Objective (\$)
& Gap (\%)
& Time (s)
& Objective (\$)
& Gap (\%)
& Time (s)\\
\midrule
25.0000 &   5
& \(-83{,}775.000\) & 11.241 &  0.100
& \(0.000\)          & 100.000 & 0.000197\\
12.5000 &   9
& \(-84{,}535.000\) & 10.436 &  0.278
& \(0.000\)          & 100.000 & 0.001544\\
6.2500  &  17
& \(-85{,}535.000\) &  9.376 &  0.581
& \(-58{,}983.125\) & 37.508 & 0.010493\\
3.1250  &  33
& \(-90{,}735.000\) &  3.867 &  1.228
& \(-62{,}171.250\) & 34.130 & 0.028680\\
1.5625  &  65
& \(-93{,}108.750\) &  1.352 &  3.006
& \(-62{,}171.250\) & 34.130 & 0.102541\\
0.7813  & 129
& \(-93{,}647.813\) &  0.781 &  8.761
& \(-64{,}466.250\) & 31.699 & 0.399467\\
0.3906  & 257
& \(-93{,}870.469\) &  0.545 & 29.354
& \(-65{,}361.250\) & 30.750 & 1.294658\\
\bottomrule
\end{tabular}
\end{table*}

Table~\ref{tab:grid-resolution} examines reservoir-grid refinement with fixed ramping grids of \(N_H=6\) for the event-based LP and \(N_H=30\) for the stagewise DP. For the event-based LP, increasing \(N_M\) from \(5\) to \(257\) monotonically reduces the gap from \(11.24\%\) to \(0.55\%\), while runtime rises from \(0.10\) to \(29.35\)~s, illustrating the accuracy--computation tradeoff and the error bound in \textbf{Proposition~\ref{prop:4}}.

Table~\ref{tab:grid-resolution} also compares the event-boundary LP with a stagewise finite-grid DP, both of which solve their respective grid-restricted models exactly. Stagewise discretization is much more restrictive, returning the all-offline solution on the two coarsest grids. At comparable runtimes, the event-boundary LP attains a \(3.867\%\) gap using \(N_M=33\) and \(N_H=6\) in \(1.228\)~s, whereas the stagewise DP retains a \(30.75\%\) gap despite using much finer grids \((N_M=257\) and \(N_H=30)\) in \(1.295\)~s. Thus, retaining continuous within-event trajectories provides substantially greater accuracy for a given computational effort.

\begin{table}[h!]
\centering
\caption{Results for the 336-stage RTS-24 instance.}
\label{tab:end-to-end}
\scriptsize
\setlength{\tabcolsep}{4pt}
\renewcommand{\arraystretch}{1.08}
\begin{tabular}{@{}lrr@{}}
\toprule
Metric & MILP & Finite-grid LP \\
\midrule
System cost (\$)       & 3,746,325.93 & 3,746,758.50 \\
System-cost gap (\%)   & 0            & 0.01       \\
CG iterations          & 334          & 313          \\
Generated columns      & 350          & 321          \\
Pricing time (s)       & 434.10       & 10.42        \\
Total time (s)         & 471.55       & 36.49        \\
Load shedding (MWh)    & 0            & 0            \\
Curtailment (MWh)      & 535.58       & 462.64       \\
\bottomrule
\end{tabular}
\end{table}

\begin{figure}[h!]
    \centering
    \includegraphics[scale=0.77]
    {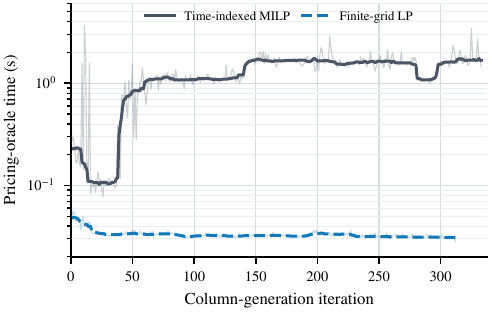}
    \caption{Per-iteration pricing time during CG.
    Light traces report individual iterations, whereas thick curves report
    15-iteration rolling medians.}
    \label{fig:pricing-time-per-iteration}
\end{figure}

We next examine whether finite-grid pricing preserves system-level solution quality while reducing computation. Table~\ref{tab:end-to-end} shows that the finite-grid LP increases system cost by only \(0.01\%\), while reducing total solution time from \(471.55\)~s to \(36.49\)~s, a \(12.9\times\) speedup. Pricing time is reduced by a factor of \(41.7\), with comparable numbers of iterations and generated columns. Thus, the finite-grid approximation has negligible impact on system cost while substantially accelerating PSH pricing.

Figure~\ref{fig:pricing-time-per-iteration} shows a widening difference in pricing performance over the CG iterations. The finite-grid LP pricing time quickly decreases to approximately \(0.03\)~s per iteration and remains nearly constant, whereas the time-indexed MILP pricing time generally increases to \(1\) to \(2\)~s. 
%This pattern is consistent with fewer event arcs remaining competitive as the dual variables stabilize, while the MILP requires increasing effort to certify near-zero reduced costs. 
%These per-iteration differences produce the cumulative pricing times reported in Table~\ref{tab:end-to-end}.

To further test the scalability of the proposed LP approximation, we scale up the above RTS-24 instance by increasing the number of PSH units from 3 to 9 and extending the horizon from \(T=336\) to \(T=672\). Table~\ref{tab:end-to-end-large} and Fig.~\ref{fig:pricing-large} demonstrate that the computational advantage of the finite-grid LP increases substantially for the larger instance. The LP converges in \(594.34\)~s despite requiring \(1{,}014\) CG iterations, whereas the MILP fails to converge within two hours after 589 iterations with a final reduced cost of $-0.78$. By that point, the MILP had accumulated \(6{,}459.12\)~s of pricing time, \(138.3\times\) the \(46.72\)~s required by the complete LP run. As CG progresses, the LP pricing time remains near \(0.04\) to \(0.05\)~s per iteration, while the MILP pricing time increases to approximately \(10\) to \(20\)~s and becomes more variable. The LP system cost is \(0.02\%\) above the MILP value available at the time limit, although this comparison is provisional as the MILP run did not converge. 

\begin{table}[!h]
\centering
\caption{Results for the 672-stage RTS-24 instance with nine PSH units.}
\label{tab:end-to-end-large}
\scriptsize
\setlength{\tabcolsep}{4pt}
\renewcommand{\arraystretch}{1.08}
\begin{tabular}{@{}lrr@{}}
\toprule
Metric & MILP & Finite-grid LP \\
\midrule
Converged                         & No             & Yes          \\
System cost (\$)                  & 7,482,570.62    & 7,484,032.04 \\
Difference from MILP value (\%)  & 0              & 0.02       \\
CG iterations                     & 589            & 1,014        \\
Generated columns                 & 2,253          & 2,495        \\
Final reduced cost                & $-0.78$      & 0             \\
Pricing time (s)                  & 6,459.12       & 46.72         \\
Total time (s)                    & 7,200 (limit)  & 594.34        \\
Load shedding (MWh)               & --             & 0             \\
Curtailment (MWh)                 & --             & 0             \\
\bottomrule
\end{tabular}

\vspace{2mm}
\parbox{\columnwidth}{\scriptsize\textit{Note:} The MILP entries are values at
the two-hour limit. Because the MILP did not converge, the
reported cost difference is provisional rather than an optimality gap.}
\end{table}

\begin{figure}[!h]
\centering
\includegraphics[scale=0.77]{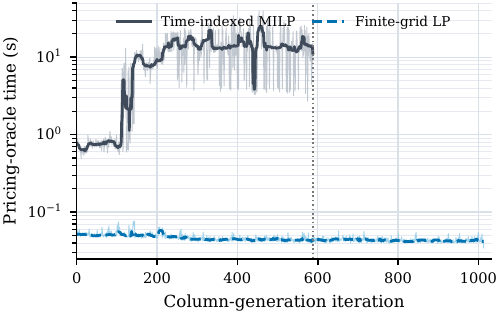}
\caption{Per-iteration pricing time during CG for the 672-stage instance. 
Light traces report individual iterations, while dark curves report 15-iteration rolling medians. 
The vertical dotted line marks the two-hour limit for the MILP run.}
\label{fig:pricing-large}
\end{figure}

We also conducted tests on modified IEEE 300-bus test-system instances and obtained results qualitatively similar to those from the RTS-24 cases. We omit these results for brevity.\looseness = -1

\section{Conclusions}
\label{sec:conc}
This paper develops an event-based framework for single-unit PSH scheduling that separates discrete mode transitions from continuous operation. We establish its equivalence to the time-indexed MILP and show that enumerating only the discrete event paths yields an ideal finite-dimensional convexification while retaining continuous reservoir and ramping states. This structure leads to a finite-grid convex hull LP with a controllable discretization-error bound and an exact event-based B\&B method used to verify the developed continuous-state theory. 
Standalone and system-level experiments demonstrate the accuracy and scalability of the finite-grid LP for repeated PSH pricing. 
%In the enlarged 672-stage RTS-24 instance, the finite-grid approach converges within ten minutes, whereas the time-indexed MILP pricing approach does not converge within the two-hour limit.
\vspace{-0.9em}
\bibliographystyle{IEEEtran}
\bibliography{ref}

\newpage 
\appendices
\section{Numerical Experiment Settings}
\label{app:numerical-settings}

\subsection{Standalone Single-Unit Experiments}
\label{sec:common-setup}

The standalone benchmark is adapted from \cite{Qu2023}, with additional parameters introduced for the four-mode formulation. 
We consider one-hour stages with \(T=24\) and the operating-mode set $\mathcal X=\{\mathsf G,\mathsf P,\mathsf{SC},\mathsf O\}$.

The implemented minimum event durations are $D_{\mathsf G}^{\min}=D_{\mathsf P}^{\min}=6$ and $D_{\mathsf{SC}}^{\min}=D_{\mathsf O}^{\min}=1$.
An event beginning near the end of the horizon may be truncated at the
terminal boundary. Generating, pumping, and offline events may span the
entire remaining horizon. HSC events are represented using blocks of at
most four stages. Consecutive HSC blocks are allowed, retain reservoir
and ramping continuity, and incur zero transition cost. Therefore, the
four-stage HSC block limit is a computational event-length limit rather
than a restriction on the total number of consecutive HSC stages.

The pumping efficiency and generation water-conversion coefficient are $\alpha=0.75$ and $\mu=0.75$.
Natural inflow and spillage are zero at every stage. The reservoir dynamics are therefore
\[
M_{t+1}
=
M_t-\mu H_t^{\mathsf O}
+\alpha H_t^{\mathsf I}.
\]
The initial and terminal boundary states are fixed as
\begin{align*}
x_1&=\mathsf O,
&
M_1&=450,
&
\bar H_1&=0,\\
x_{T+1}&=\mathsf O,
&
M_{T+1}&=450,
&
\bar H_{T+1}&=0.
\end{align*}
Thus, the unit begins and terminates offline with the same reservoir level and a zero terminal turbine-output boundary. 
The reservoir capacity is $\overline M=900$.

Turbine generation is bounded by $40\le H_t^{\mathsf O}\le130$ in modes \(\mathsf G\) and \(\mathsf{SC}\). Pumping input satisfies $0\le H_t^{\mathsf I}\le130$
in mode \(\mathsf P\), while $52\le H_t^{\mathsf I}\le130$ in mode \(\mathsf{SC}\). 
The strictly positive HSC pumping minimum prevents HSC from reducing to a duplicate generating mode. 
The turbine ramp-up and ramp-down limit is $\overline V=50$.
No intertemporal ramping constraint is imposed on pumping input.

The startup, shutdown, and online mode-change costs are $K^{\mathrm{SU}}=1200$, $K^{\mathrm{SD}}=300$, and $K^{\mathrm{MC}}=0$, respectively. 
The transition cost is
\[
\Gamma(x,x')
=
\begin{cases}
K^{\mathrm{SU}},
& x=\mathsf O,\quad x'\ne\mathsf O,\\
K^{\mathrm{SD}},
& x\ne\mathsf O,\quad x'=\mathsf O,\\
K^{\mathrm{MC}},
& x,x'\ne\mathsf O,\quad x\ne x',\\
0,
& x=x'.
\end{cases}
\]
The terminal transition into \(x_{T+1}=\mathsf O\) is included. In
particular, a transition between consecutive HSC blocks has zero cost.

The mode-specific fixed operating costs are $\kappa^{\mathsf G}=120$, $\kappa^{\mathsf P}=80$, $\kappa^{\mathsf{SC}}=220$, and $\kappa^{\mathsf O}=0$,
and the variable generation and pumping costs are
$c^{\mathsf G}=2\ \$/\mathrm{MWh}$ and $c^{\mathsf P}=0.5\ \$/\mathrm{MWh}$.
Each method minimizes the net operating cost
\[
\sum_{t=1}^{T}
\left[
\bigl(c^{\mathsf G}-\lambda_t\bigr)H_t^{\mathsf O}
+
\bigl(c^{\mathsf P}+\lambda_t\bigr)H_t^{\mathsf I}
+
\kappa^{x_t}
+
\Gamma(x_t,x_{t+1})
\right].
\]
The corresponding operating profit is the negative of this objective value. 
Table~\ref{tab:price-vector} reports the 24-stage electricity-price profile used in the code.
\begin{table}[!t]
\centering
\scriptsize
\setlength{\tabcolsep}{2.5pt}
\renewcommand{\arraystretch}{0.95}
\caption{Standalone electricity prices \(\lambda_t\), in \$/MWh.}
\label{tab:price-vector}
\begin{tabular}{c*{6}{c}|c*{6}{c}}
\toprule
\(t\) & 1 & 2 & 3 & 4 & 5 & 6
& \(t\) & 13 & 14 & 15 & 16 & 17 & 18\\
\(\lambda_t\)
& 130 & 150 & 160 & 135 & 150 & 150
& \(\lambda_t\)
& 120 & 125 & 130 & 160 & 260 & 280\\
\midrule
\(t\) & 7 & 8 & 9 & 10 & 11 & 12
& \(t\) & 19 & 20 & 21 & 22 & 23 & 24\\
\(\lambda_t\)
& 220 & 240 & 250 & 220 & 200 & 130
& \(\lambda_t\)
& 250 & 220 & 200 & 160 & 130 & 130\\
\bottomrule
\end{tabular}
\end{table}
For the grid-restricted comparisons in Table~\ref{tab:exactness}, the boundary grids are
$\mathcal M = \{0,18,36,\ldots,900\}$ and $\mathcal H =
\{0,40,66,92,118,130\}$.
Only event-boundary reservoir and turbine-output values are
discretized; the within-event dispatch and reservoir trajectories
remain continuous. The continuous-state event-based B\&B method and
the time-indexed MILP do not impose these grids. For the
grid-restricted B\&B row, the root event-network relaxation is
integral, so optimality is certified at the root without branching.

For the grid-resolution study in
Table~\ref{tab:grid-resolution}, the ramping grid \(\mathcal H\) is
held fixed. The reservoir-grid step, expressed as a percentage of
\(\overline M\), is selected from
\[
\{25,\ 12.5,\ 6.25,\ 3.125,\ 1.5625,\ 0.78125,\ 0.390625\}\%.
\]
These resolutions produce $5,\ 9,\ 17,\ 33,\ 65,\ 129,\ \text{and }257$ reservoir boundary points, respectively. 
The initial and terminal reservoir levels are inserted whenever they are not generated directly
by the specified step.

All standalone implementations use the same physical and economic
parameters. The continuous within-event LPs and time-indexed MILPs are
solved using Gurobi~12.0. Unless a command-line time or node limit is
specified, no such limit is imposed.

\subsection{System-Level Column-Generation Experiments}
\label{sec:end-to-end-setting}

The system-level experiments use a modified IEEE RTS-24 instance with 38 transmission branches, 32 thermal generators, and a 100-MVA base. 
Bus 13 is the reference bus. 
The bus peak loads, branch data, generator capacities, and quadratic variable-cost coefficients are taken from the RTS-24 data (see, e.g., \cite{Zimmerman2011MATPOWER}). 
Thermal outputs are continuous with zero minimum output; thermal commitment, fixed costs, and constant cost terms are omitted. 
Each quadratic variable cost is approximated by four equal-width convex piecewise-linear segments. 
The ramp limit of thermal generator \(g\) is
\[
V_g
=
\min\left\{
\overline P_g,\,
\max\{20,\,0.5\overline P_g\}
\right\}
\quad\mathrm{MW/h}.
\]
Load shedding is permitted only as a feasibility safeguard and is
penalized at \(10{,}000\ \$/\mathrm{MWh}\).

Let $h(t)=1+\bigl((t-1)\bmod24\bigr)$ and $d(t)=1+\left(\left\lfloor\frac{t-1}{24}\right\rfloor\bmod7\right)$.
The hourly load-shape vector and seven-day load factors are
\begin{align*}
q^{D}
={}&(0.67,0.63,0.60,0.59,0.59,0.60,0.74,0.86,\\
&\quad 0.95,0.96,0.96,0.95,0.95,0.95,0.93,0.94,\\
&\quad 0.99,1.00,1.00,0.96,0.91,0.83,0.73,0.63),\\
\eta^{D}
={}&(1.00,0.96,1.02,1.00,0.98,0.93,0.90).
\end{align*}
Thus, the demand at bus \(n\) and stage \(t\) is $D_{n,t} = D_n^{\mathrm{peak}}q^D_{h(t)}\eta^D_{d(t)}$.

Three wind farms are located at buses 3, 15, and 21, with nominal
capacities of 250, 300, and 350~MW and profile shifts of 0, 7, and
13 hours, respectively. The base wind-capacity-factor profile and
seven-day factors are
\begin{align*}
q^{W}
={}&(0.70,0.74,0.76,0.78,0.80,0.82,0.75,0.62,\\
&\quad 0.50,0.42,0.38,0.35,0.32,0.30,0.33,0.40,\\
&\quad 0.48,0.58,0.65,0.72,0.78,0.80,0.77,0.73),\\
\eta^{W}
={}&(1.00,0.85,1.10,0.95,1.05,0.90,1.15).
\end{align*}
For wind farm \(w\), define
\[
h_w(t)=1+\bigl((t-1+s_w)\bmod24\bigr),
\]
where \(s_w\) is its profile shift. Its available output is
\[
A_{w,t}
=
1.25\,C_w
\min\left\{0.95,\,
q^W_{h_w(t)}\eta^W_{d(t)}
\right\}.
\]
The seven-day load and wind factors repeat when the horizon exceeds one
week.

The base experiment has \(T=336\) stages and three identical PSH units located at buses $\mathcal B_{336}^{\mathrm{PSH}} = \{7,15,23\}$.
The enlarged experiment has \(T=672\) stages and nine identical PSH units located at $\mathcal B_{672}^{\mathrm{PSH}} = \{2,3,7,10,13,15,18,21,23\}$.

Both system experiments use generating, pumping, HSC, and offline modes. The common PSH parameters are
\begin{align*}
[\underline B^{\mathsf O},\bar B^{\mathsf O}]
&=[40,130]\ \mathrm{MW},
&
[\underline B^{\mathsf I},\bar B^{\mathsf I}]
&=[0,130]\ \mathrm{MW},\\
\overline V&=50\ \mathrm{MW/h},
&
\overline M&=900\ \mathrm{MWh},\\
\alpha&=0.75,
&
\mu&=1,\\
L&=2,
&
l&=2,\\
K^{\mathrm{SU}}&=\$800,
&
K^{\mathrm{SD}}&=\$0.
\end{align*}
Natural inflow and spillage are zero. 
Each unit begins the first stage offline and satisfies $M_1=M_{T+1}=450$. 
No terminal operating-mode or terminal turbine-output condition is imposed in the system experiments. 
The physical generation and pumping costs are linear and equal to \(5\ \$/\mathrm{MWh}\) and \(2\ \$/\mathrm{MWh}\), respectively. 
Transitions between generating and pumping modes have zero cost.

For the finite-grid pricing oracle, the boundary grids are
\begin{align*}
\mathcal M&=\{0,150,300,450,600,750,900\},\\
\mathcal H&=\{0,40,90,130\}.
\end{align*}
Event blocks contain at most four stages, but consecutive blocks in the same mode share boundary states at zero transition cost, permitting longer operating runs. 
The finite-grid oracle reuses its network and within-event feasible regions, updating only dual-dependent objective coefficients during CG. 
Pumping dispatch is computed analytically, generation-event costs use cached extreme points, and pricing is solved by shortest path.

The restricted master combines complete PSH trajectories, enforces network constraints, and requires each unit's column weights to sum to one. 
Both implementations share the master formulation, system data, all-offline initialization, and stopping rule, differing only in continuous-state time-indexed MILP or finite-grid LP pricing. 
CG terminates when no column is added and the minimum reduced cost across units is at least \(-10^{-6}\).

For \(T=672\), finite-grid CG satisfies the stopping criterion, whereas MILP CG is terminated after two hours. Its reported values correspond to the last completed iteration and do not represent convergence.

Pricing time accumulates wall-clock time across all pricing iterations and PSH units. Total time additionally includes model and oracle construction and all restricted-master solves.

\section{Proofs}
\label{app:finite_dimensional_convexification}
This appendix proves \textbf{Propositions~1}--\textbf{5}.

\begin{proof}[Proof of Proposition 1]
We establish the two correspondences stated in the proposition.

\medskip
\noindent
\textbf{From the time-indexed MILP to the event formulation.}
Consider a feasible time-indexed solution. Its mode at stage \(t\) is
\[
x_t=
\begin{cases}
\mathsf G, & y_t^{\mathsf G}=1,\\
\mathsf P, & y_t^{\mathsf P}=1,\\
\mathsf{SC}, & y_t^{\mathsf{SC}}=1,\\
\mathsf O, & y_t=0.
\end{cases}
\]
Constraint~\eqref{eq:status} and the binary restrictions ensure that
exactly one case applies. Set \(x_{T+1}=x_T\), and define
\[
\bar H_1=0,
\qquad
\bar H_{t+1}=
\begin{cases}
H_t^{\mathsf O}, & x_t\in\{\mathsf G,\mathsf{SC}\},\\
0, & x_t\in\{\mathsf P,\mathsf O\}.
\end{cases}
\]
Starting from \(\tau_1=0\), update \(\tau_{t+1}\) by
\eqref{eq:tau}. For \(t=1,\ldots,T\), consider the one-stage event
\[
\begin{aligned}
s_t&=(x_t,M_t,\bar H_t,\tau_t),\\
e_t&=(t+1,x_{t+1},M_{t+1},\bar H_{t+1}).
\end{aligned}
\]
The minimum up/down-time constraints imply that an online/offline
transition occurs only after the compulsory residence time represented
by \(\tau_t\) has elapsed. The reservoir and output bounds place
\(M_{t+1}\) and \(\bar H_{t+1}\) in their state domains, while the
ramping constraint supplies the required ramp-down condition whenever
the successor mode has no turbine output. Hence
\(e_t\in\mathcal E(s_t)\).

Restrict the time-indexed continuous variables to stage \(t\). If
\(x_t=\mathsf G\), constraints
\eqref{eq:gen-bounds}--\eqref{eq:ramp-up} and the reservoir balance
reduce to the one-stage generating block. If \(x_t=\mathsf P\), they
reduce to the pumping block. If \(x_t=\mathsf{SC}\), both
\(H_t^{\mathsf O}\) and \(H_t^{\mathsf I}\) are active and the
constraints reduce to the HSC block. Finally, if \(x_t=\mathsf O\),
both power variables are zero and the offline reservoir equation
applies. The definitions of \(\bar H_t\) and \(\bar H_{t+1}\) also
enforce the incoming and outgoing ramping boundaries, including the
ramp-down condition when the next mode has no turbine output.

For a one-stage event, its boundary data fix the inherited stagewise
dispatch: \(\bar H_{t+1}\) fixes \(H_t^{\mathsf O}\) whenever turbine
output is active, and the reservoir balance fixes
\(H_t^{\mathsf I}\) whenever pumping input is active. The epigraph
constraints in the corresponding block therefore give
\[
c_{t,t+1}(s_t,e_t)
=
h_t\!\left(
H_t^{\mathsf O},H_t^{\mathsf I},
y_t^{\mathsf G},y_t^{\mathsf P},y_t^{\mathsf{SC}}\right).
\]
Moreover, for \(t=1,\ldots,T-1\),
\[
\Gamma_{t,t+1}(x_t,x_{t+1})
=SU_{t+1}u_{t+1}+SD_{t+1}d_{t+1}.
\]
Because \(y_0=y_1=0\), we have \(u_1=d_1=0\), and the choice
\(x_{T+1}=x_T\) creates no terminal transition cost. Summing the
one-stage event costs therefore reproduces the MILP objective exactly.
Thus every feasible MILP solution induces a feasible event schedule
with the same cost, and
\[
V_1(s_1)\leq \min_{q\in\mathcal Y}C(q).
\]

\medskip
\noindent
\textbf{From the event formulation to the time-indexed MILP.}
Conversely, consider a feasible event schedule
\[
\begin{aligned}
s_{t_1}\xrightarrow{e_{t_1}}s_{t_2}
&\xrightarrow{e_{t_2}}\cdots
\xrightarrow{e_{t_K}}s_{t_{K+1}},
\end{aligned}
\]
with $t_1=1$ and $t_{K+1}=T+1$, and choose an optimal solution of every within-event LP. Concatenate these solutions. For \(i=t_k,\ldots,t_{k+1}-1\), set
\[
(y_i^{\mathsf G},y_i^{\mathsf P},y_i^{\mathsf{SC}})
=
\begin{cases}
(1,0,0), & x_{t_k}=\mathsf G,\\
(0,1,0), & x_{t_k}=\mathsf P,\\
(0,0,1), & x_{t_k}=\mathsf{SC},\\
(0,0,0), & x_{t_k}=\mathsf O,
\end{cases}
\]
and define \(y_i\), \(g_i\), and \(p_i\) as in Section \ref{sec:MILP}. Define \(u_i=1\) at an offline-to-online transition,
\(d_i=1\) at an online-to-offline transition, and set them to zero
otherwise.

Within each event, the applicable mode-specific block gives the
time-indexed capacity, reservoir, ramping, and cost constraints.
At a common event boundary, the outgoing reservoir and ramping values
of one event are the incoming values of the next. Thus the
concatenated reservoir trajectory is continuous. The ramping
constraint also holds across the boundary: two consecutive
turbine-active events share the same ramping boundary; entry from
\(\mathsf P\) or \(\mathsf O\) uses a zero boundary; and exit to
\(\mathsf P\) or \(\mathsf O\) is feasible only when the preceding
terminal turbine output is at most \(\overline V\).

Finally, \eqref{eq:tau} and the event-feasibility condition enforce the minimum up/down times, while \(\Gamma_{t_k,t_{k+1}}(x_{t_k},x_{t_{k+1}})\) records exactly the start-up or shut-down cost at the boundary. 
Hence the concatenated variables form a feasible time-indexed solution with the same objective value as the event schedule. 
It follows that $\min_{q\in\mathcal Y}C(q)\leq V_1(s_1)$. Combining the two inequalities proves the proposition.
\end{proof}

\begin{proof}[Proof of Proposition 2]
We first give a compact form of the finite-dimensional pathwise LP.
It is algebraically equivalent to concatenating the mode-specific
event blocks along each complete discrete event path, but it avoids
introducing duplicate copies of the reservoir and ramping boundaries.

\medskip
\noindent
\textbf{Convex-hull equality.}
Consider a feasible point of \(\mathcal P^{\mathsf{FD}}\). For every
path with \(\rho_{\tilde\omega}>0\), divide all its scaled physical
variables by \(\rho_{\tilde\omega}\). Dividing
\eqref{eq:FD-path-physical} by this weight gives the original
time-indexed physical constraints with the mode and transition
variables fixed by \(\tilde\omega\). Equivalently, it concatenates the
corresponding generating, pumping, HSC, and offline event blocks.
Hence the normalized variables define a schedule
\(q^{\tilde\omega}\in\mathcal Y\). If
\(\rho_{\tilde\omega}=0\), the scaled physical bounds force all
physical variables of that path to zero. The reconstruction equations
therefore give $\bar q
=
\sum_{\substack{\tilde\omega\in\widetilde\Omega}}
\rho_{\tilde\omega}q^{\tilde\omega}$ and $ \sum_{\substack{\tilde\omega\in\widetilde\Omega}}
\rho_{\tilde\omega}=1$ for $\rho_{\tilde\omega}>0$.
Thus $\operatorname{proj}_{\bar q} \left(\mathcal P^{\mathsf{FD}}\right) \subseteq \operatorname{conv}(\mathcal Y)$.

Conversely, let
\(\bar q=\sum_{r=1}^{R}\beta_rq^r\), where
\(q^r\in\mathcal Y\), \(\beta_r\geq0\), and
\(\sum_{r=1}^{R}\beta_r=1\). By Proposition~\ref{prop:1}, each \(q^r\) induces a
complete path \(\tilde\omega(r)\in\widetilde\Omega\). Choose cost
epigraph variables attaining \(C(q^r)\), and set
\[
\rho_{\tilde\omega}
:=
\sum_{\{r:\tilde\omega(r)=\tilde\omega\}}\beta_r.
\]
For each path, define every scaled physical and cost variable as the
corresponding \(\beta_r\)-weighted sum over schedules assigned to that
path. All schedules assigned to the same path satisfy the same linear
system. Their weighted sum therefore satisfies
\eqref{eq:FD-path-physical}--\eqref{eq:FD-path-total-cost} and
reconstructs \(\bar q\). Hence
\[
\operatorname{conv}(\mathcal Y)
\subseteq
\operatorname{proj}_{\bar q}
\left(\mathcal P^{\mathsf{FD}}\right).
\]
The two inclusions prove the stated results.

\medskip
\noindent
\textbf{Optimal values and path concentration.}
For every positive-weight path, normalize its cost variables as well and let \(\widehat C^{\tilde\omega}\) denote the resulting path cost.
Because the variables form an epigraph, $\widehat C^{\tilde\omega} \geq C(q^{\tilde\omega})$.
For a zero-weight path, \eqref{eq:FD-path-physical}--\eqref{eq:FD-path-cost-epigraph} imply
\(\widetilde C^{\tilde\omega}\geq0\). 
Therefore every feasible point satisfies
\begin{align*}
\bar C
&\geq
\sum_{\substack{\tilde\omega\in\widetilde\Omega\\
\rho_{\tilde\omega}>0}}
\rho_{\tilde\omega}C(q^{\tilde\omega})\\
&\geq
\min_{q\in\mathcal Y}C(q).
\end{align*}
Conversely, any \(q\in\mathcal Y\) can be represented by assigning
weight one to an induced path, weight zero to all other paths, and
choosing the cost epigraph variables at their minimum feasible values.
Consequently,
\[
\min_{(\bar q,\cdot)\in\mathcal P^{\mathsf{FD}}}\bar C
=
\min_{q\in\mathcal Y}C(q)
=
V_1(s_1),
\]
where the last equality follows from Proposition~\ref{prop:1}. In particular,
if \(q^\star\) is an optimal deterministic schedule with path
\(\tilde\omega^\star\), assigning
\(\rho_{\tilde\omega^\star}=1\) and zero weight to every other path
gives an optimal solution of the finite-dimensional LP. This proves
the proposition.
\end{proof}

\begin{proof}[Proof of Proposition 3]
As in the unit-flow description preceding \eqref{eq:ALLINONE}, the
initial layer has the single source \(\hat s_1\); equivalently, every
other stage-1 outflow is fixed to zero. This is the standard
single-source convention for the displayed flow model.

Abbreviate an arc \((\hat s_t,\hat a_t)\) by \(a\), its flow by
\(\pi_a\), and its scaled within-event variables by \(\tilde z_a\).
Let \(\mathcal P_a\) be the unscaled physical feasible set of the
generating, pumping, HSC, or offline block attached to \(a\). The
boundary values in \(\mathcal P_a\) are fixed by the labels of the
initial and terminal grid nodes. The convex-hull claim concerns the
flow and physical schedule variables, so the cost-epigraph variables
are projected out; their upward slack affects neither the physical
schedule nor the optimal value.

The arc-embedded constraints in \eqref{eq:ALLINONE} are precisely the
perspective-scaled constraints $\tilde z_a\in \pi_a\mathcal P_a$.
In particular, if \(\pi_a>0\), then
\(z_a:=\tilde z_a/\pi_a\in\mathcal P_a\); if \(\pi_a=0\), the scaled
physical bounds force \(\tilde z_a=0\).

Consider any feasible solution of \eqref{eq:ALLINONE}. Because every
arc advances the stage index, its unit flow admits a source-to-sink
path decomposition
\[
\pi=\sum_{r=1}^{R}\rho_r\chi^{(r)},
\qquad
\rho_r\geq0,
\qquad
\sum_{r=1}^{R}\rho_r=1,
\]
where \(\chi^{(r)}\) is the incidence vector of a path in the finite-grid event network. For each path \(r\), assign \(z_a\) to every
arc \(a\) on that path. These arcwise decisions form a feasible
integral event-network solution: every local block is feasible, and
adjacent arcs have common boundary values because they meet at the
same grid node. Moreover, for every active arc, $\sum_{\{r:a\in\chi^{(r)}\}}\rho_r z_a =\pi_a z_a =\tilde z_a$.
Thus every point in the projection of \eqref{eq:ALLINONE} is a convex
combination of feasible integral event-network points.

For the reverse inclusion, take any convex combination of feasible
integral event-network points. Set each \(\pi_a\) equal to the total
weight of paths using \(a\), and set \(\tilde z_a\) equal to the
corresponding weighted sum of their local variables. Flow conservation
holds by linearity. If \(\pi_a>0\), then
\(\tilde z_a/\pi_a\) is a convex combination of points in the
polyhedron \(\mathcal P_a\), and hence belongs to \(\mathcal P_a\).
Therefore every scaled block is feasible. This proves that
\eqref{eq:ALLINONE} is an extended convex-hull formulation of the
finite-grid event-network model.

It remains to verify the claim about an integral optimum. 
Normalize the full local variables, including their cost epigraph variables, on every arc with \(\pi_a>0\), and let \(\ell_a\) be the resulting local operating cost plus the transition cost \(\Gamma(\hat s_t,\hat a_t)\).
Redundant cost slack on inactive arcs can be set to zero. 
The objective of the feasible LP solution is then $\sum_a\pi_a\ell_a = \sum_{r=1}^{R}\rho_r
\sum_{a\in\chi^{(r)}}\ell_a$.
At least one path has cost no greater than this weighted average.
Selecting that path integrally and using its normalized local solutions gives a feasible integral solution with no larger objective value; reoptimizing its local blocks can only improve it. Applying this construction to an optimal LP solution proves that an integral optimal solution exists. 
Such a solution is exactly an optimal path of the finite-grid DP.
\end{proof}

\begin{proof}[Proof of Proposition 4]
Write \(\xi_t=(M_t,\bar H_t)\) and let $\delta:=\Delta_M+\Delta_H$.
The feasible-grid-perturbation condition is understood in the event-network sense: when two consecutive events share a boundary, the same perturbed grid value is used as the outgoing boundary of the first event and the incoming boundary of the second. 
Otherwise, the perturbed events would not form a feasible grid path.

For fixed discrete event data \((t,j,x_t,x^\dagger,\tau_t)\), the within-event problem is a parametric LP whose parameters are the incoming and outgoing boundary values. 
Its physical variables are bounded, and those boundary values enter linearly. 
Consequently, its optimal value is a polyhedral, piecewise-affine function on its feasible boundary domain and is Lipschitz there. 
Since the horizon and the mode set are finite, one constant \(L\), independent of the grid
resolution, can be chosen for all feasible event labels and all four mode-specific blocks. 
Thus, whenever two events have the same discrete data and feasible boundary pairs
\((\xi_t,\xi_j)\) and \((\xi_t',\xi_j')\),
\[
\begin{aligned}
\left|
c_{t,j}(\xi_t,\xi_j)
-c_{t,j}(\xi_t',\xi_j')
\right| \leq
L\left(
\|\xi_t-\xi_t'\|_1+\|\xi_j-\xi_j'\|_1
\right).
\end{aligned}
\]

Let \(z^\star=V_1(s_1)\) be the continuous-state optimal value, and fix an optimal event path with boundaries $1=t_1<t_2<\cdots<t_{K+1}=T+1$.
By the feasible-grid-perturbation condition in the proposition, its shared boundary at each \(t_k\) can be replaced by a common grid boundary \(\hat\xi_{t_k}\) so that every perturbed event remains feasible, its switching time and successor mode are unchanged, and $\|\hat\xi_{t_k}-\xi_{t_k}\|_1 \leq C_B\delta$.
The same perturbed value is used as the outgoing boundary of event \(k-1\) and the incoming boundary of event \(k\). 
Hence the perturbed events form a feasible path of the finite-grid DP.

The transition costs are unchanged because the discrete event data
are unchanged. The Lipschitz bound gives, for every event on the path,
\[
\left|
c_{t_k,t_{k+1}}(\hat\xi_{t_k},\hat\xi_{t_{k+1}})
-c_{t_k,t_{k+1}}(\xi_{t_k},\xi_{t_{k+1}})
\right|
\leq 2LC_B\delta.
\]
Since \(K\leq T\), the cost of this grid-feasible path is at most $z^\star+2LTC_B\delta$.
Let \(\hat z^\star\) be the optimal value of the finite-grid DP.
Every grid-feasible path is also feasible for the continuous-state DP, so \(z^\star\leq\hat z^\star\). 
The preceding grid path gives the reverse estimate $\hat z^\star\leq z^\star+2LTC_B\delta$.
Therefore $\left|\hat z^\star-z^\star\right| \leq C(\Delta_M+\Delta_H)$, where $C:=2LTC_B$ is independent of the grid resolution.
\end{proof}

\begin{proof}[Proof of Proposition 5]
For a complete discrete event path
\(\tilde\omega\in\widetilde\Omega\), let \(z(\tilde\omega)\) be the
optimal value of its continuous path LP, with
\(z(\tilde\omega)=+\infty\) if that LP is infeasible. By the
event-based reformulation, we have $V_1(s_1) = \min_{\tilde\omega\in\widetilde\Omega} z(\tilde\omega)$.

At a B\&B node \(n\), let \(\widetilde\Omega(n)\) be the set of
complete paths extending its fixed prefix \(P_n\). Every feasible
solution of the path LP for a path in \(\widetilde\Omega(n)\) induces
an integral unit flow satisfying the fixed-prefix constraints. For
that integral flow, the lifted boundary variables equal the intended
products, so they satisfy the McCormick inequalities
\eqref{eq:ebb-mccormick}. The pathwise boundary equalities imply
\eqref{eq:ebb-lifted-conservation}, and the perspective-scaled
generating, pumping, HSC, and offline blocks hold with the same
objective value. Hence every exact completion of \(n\) is feasible in
the node relaxation, and
\[
\mathrm{LB}(n)
\leq
\min_{\tilde\omega\in\widetilde\Omega(n)}
z(\tilde\omega).
\]
Thus pruning a node with \(\mathrm{LB}(n)\geq\mathrm{UB}\) cannot
discard a completion that improves the incumbent. Conversely, every
incumbent is obtained by solving a complete path LP exactly and is
therefore the objective value of a feasible event schedule.

The search tree is finite. Indeed, at stage \(t\), both
\(j\in\{t+1,\ldots,T+1\}\) and \(x^\dagger\in\mathcal X\) have
finitely many choices, and the stage index strictly increases along
each branch. Hence the depth is at most \(T\), and exhaustive
branching over all unpruned nodes terminates after finitely many
steps.

Let \(\tilde\omega^\star\) be an optimal complete path. If its leaf is
evaluated, the incumbent equals the global optimum no later than that
evaluation. Otherwise, some ancestor \(n\) of this leaf is pruned.
At the time of pruning,
\[
\begin{aligned}
V_1(s_1)
&\leq \mathrm{UB}
\leq \mathrm{LB}(n)\\
&\leq
\min_{\tilde\omega\in\widetilde\Omega(n)}
z(\tilde\omega)\\
&\leq z(\tilde\omega^\star)
=V_1(s_1).
\end{aligned}
\]
All inequalities must therefore be equalities, so the incumbent is
already optimal when that ancestor is pruned. In either case, the
final incumbent equals \(V_1(s_1)\). This proves finite termination
and exactness.
\end{proof}

\end{document}